\documentclass[12pt]{article}

\usepackage[numbers]{natbib}

\usepackage[ruled,vlined]{algorithm2e}
\usepackage{boxedminipage}
\usepackage[most]{tcolorbox}
\usetikzlibrary{shadows}
\usepackage{amsfonts}
\usepackage{amsthm}
\usepackage{tikz}
\usetikzlibrary{shapes.geometric}
\usetikzlibrary{arrows.meta}
\usetikzlibrary{positioning}
\usetikzlibrary{arrows}
\usepackage{indentfirst}
\usepackage{amsmath}
\usepackage[a4paper,margin=1.8cm]{geometry}
\usepackage{hyperref}
\usepackage{color}
\usepackage{graphicx}
\usepackage{float}
\usepackage{capt-of}

\newcommand{\br}[1]{\mathopen{}\left( #1 \right)}
\newcommand{\brc}[1]{\mathopen{}\left\{ #1 \right\}}
\newcommand{\spr}[1]{\mathopen{}\left| #1 \right|}
\newcommand{\fl}[1]{\mathopen{}\left\lfloor #1 \right\rfloor}
\newcommand{\cl}[1]{\mathopen{}\left\lceil #1 \right\rceil}
\newcommand{\angl}[1]{\mathopen{}\langle #1 \rangle}

\newcommand{\diam}{\operatorname{diam}}

\newcommand{\OPT}{\texttt{OPT}}

\SetKwFunction{FSeparator}{Separator}
\SetKwFunction{FSeparatorFPTAS}{SeparatorFPTAS}
\SetKwFunction{FAlgorithmMinCut}{AlgorithmMinCut}

\theoremstyle{plain}
\newtheorem{theorem}{Theorem}[section]
\newtheorem{lemma}[theorem]{Lemma}

\theoremstyle{definition}

\theoremstyle{remark}

\begin{document}
\title{Average Case Graph Searching in Non-Uniform Cost Models}

\author{%
Michał Szyfelbein\\[0.5ex]
\normalsize
\parbox{0.78\textwidth}{\centering
Department of Electronics, Telecommunications and Informatics\\
Gdańsk University of Technology\\
Gdańsk, Poland\\
\texttt{michal.szyfelbein@pg.edu.pl}}}

\date{}

\maketitle

\begin{abstract}
We consider the following generalization of the classic Binary Search Problem: a searcher is required to find a hidden target vertex $x$ in a graph $G$, by iteratively performing queries about vertices. A query to $v$ incurs a cost $c(v, x)$ and responds whether $v=x$ and if not, returns the connected component in $G-v$ containing $x$. The goal is to design a search strategy that minimizes the average-case search cost.

Firstly, we consider the case when the cost of querying a vertex is independent of the target. We develop a $\br{4+\epsilon}$-approximation FPTAS for trees running in $O(n^4/\epsilon^2)$ time and an $O({\sqrt{\log n}})$-approximation for general graphs. Additionally, we give an FPTAS parametrized by the number of non-leaf vertices of the graph. On the hardness side we prove that the problem is NP-hard even when the input is a tree with bounded degree or bounded diameter.

Secondly, we consider trees and assume $c(v, x)$ to be a monotone non-decreasing function with respect to $x$, i.e.\ if $u \in P_{v, x}$ then $c(u, x) \leq c(v, x)$. We give a $2$-approximation algorithm which can also be easily altered to work for the worst-case variant. This is the first constant factor approximation algorithm for both criterions. Previously known results only regard the worst-case search cost and include a parametrized PTAS as well as a $4$-approximation for paths.

At last, we show that when the cost function is an arbitrary function of the queried vertex and the target, then the problem does not admit any constant factor approximation under the UGC, even when the input tree is a star.
\end{abstract}


\section{Introduction}
Searching in graph structures is a fundamental problem in computer science, with applications ranging from machine learning to operations research. The input graph $G$ is assumed to contain a hidden \textit{target} element $x$ which the searcher is required to locate. In the classical setting, while exploring the graph, the searcher usually obtains only a local information about the placement of $x$. This principle underlines the well-known BFS and DFS strategies, which allow the searcher to locate the target when the only information received, while visiting a node, is whether it is the target. We study a more global search model in which the searcher is allowed to perform arbitrary \textit{queries}, each about a chosen vertex $v$. A \textit{response} to such query consists of information whether $v=x$, and if not, which of the connected components of $G-v$ contains $x$. 
Each answer allows the searcher to narrow the subset of vertices consistent with previous responses, and have not yet been eliminated as possible locations of $x$. 
We call such a set a \textit{candidate} set and the connected subgraph of $G$ induced by it a \textit{candidate} subgraph. 
This process continues until the position of the target is revealed. We wish to minimize the amount of queries needed to find $x$. When the input graph is a path, this is equivalent to searching in a linearly ordered set. In case of trees, the problem was studied among multiple names including: Tree Search Problem \cite{Cicalese2014ImprovedApproxAvgTs,Cicalese2016OnTSPwNonUniCost}, Search Trees on Trees \cite{SplayTonT,Fast_app_centroid_trees}, Hub Labeling \cite{Angelidakis2018ShortestPQ} and more. For general graphs the problem is known as Vertex Ranking \cite{RankingsofGraphs,Schaffer1989OptNodeRankOfTsInLinTime,OnakParys2006GenOfBSSInTsAndFLikePosets,Mozes_Onak2008FindOptTSStartInLinTime}, Elimination Trees \cite{Pothen1988OptimalEliminationTrees} and Hierarchical Clustering \cite{Acostfunctionforsimilaritybasedhierarchicalclustering,HCObjFsandAlgs,Approximatehierarchicalclusteringviasparsestcutandspreadingmetrics}.

In order to efficiently locate the target, the searcher needs a \textit{strategy} of searching which allows them to locate the target efficiently. The strategy is an \textit{adaptive} algorithm, which given previous responses in finite (and polynomial) time outputs the next query to be performed. We model this strategy as a rooted tree which we will call a \textit{decision tree} $D$, whose nodes are vertices of $G$. We will require that each edge outgoing from the root $r$ of $D$ is associated with a unique response to a query to $r$ in $G$ and that the same holds for all decision subtrees of $D-r$ (in the respective components of $G-v$). When searching according to $D$, the root $r$ of $D$ is queried first. If $r\neq x$ the searcher moves down along the edge of $r$ associated with the response, thereby entering the subtree $D_u$ rooted at the next queried vertex $u$. This process recurses until the target is found.

Let $c\colon V\br{G}\times V\br{G} \to \mathbb{N}$ be the \textit{query cost} function, where $c(v, x)$ represents the cost of querying $v$ when $x$ is the target\footnote{Note that the searcher is not aware of the cost incurred by the query during the search process.}. Moreover, let $w\colon V\br{G}\to \mathbb{N}$ be the \textit{weight} function, which denotes how frequently 
the vertex is searched for. For a target vertex $x \in V\br{G}$ and a decision tree $D$, let $Q_G\br{D,x}$ denote the sequence of queries performed 
along the unique path in $D$ from the root $r\br{D}$ to $x$.
The \textit{cost} of $D$ on $G$ when the target is $x$ is defined as:
$$
c_G\br{D,x} = \sum_{q\in Q_G\br{D, x}}c\br{q, x}.
$$

The \textit{weighted average-case cost} of $D$ on $G$ is defined as:
$$
c_{G, sum}\br{D} = \sum_{x\in V\br{G}}w\br{x}\cdot c_G\br{D,x}.
$$

The \textit{worst-case cost} of $D$ on $G$ is defined as:
$$
c_{G,max}\br{D} = \max_{x \in V\br{G}} c_G\br{D,x}.
$$

Hereby, we are mostly concerned with the former of these two cost measures,
and when the context is clear we omit the subscript $max$ or $sum$. We define the \textit{Graph Search Problem} as follows:
\begin{tcolorbox}[colback=white, title=Graph Search Problem (GSP), fonttitle=\bfseries, breakable]
\textbf{Input:} Graph $G$, a query cost function $c\colon V\to \mathbb{N}$ and a weight function $w\colon V\to \mathbb{N}$.

\textbf{Output:} A decision tree $D$, minimizing the weighted average case cost: 
$$
c_G\br{D} = \sum_{x\in V\br{G}}w\br{x}\cdot c_G\br{D,x}
$$
\end{tcolorbox}

We will be interested in three special cases of the query cost function:
\begin{itemize}
    \item \textbf{Costs nondependent on target:} For every $u,v \in V\br{G}$, $c\br{u,v}=c\br{u}$ .
    \item \textbf{Monotone non-decreasing costs:} For every $u,v,x \in V\br{G}$, such that $u \in P_{v,x}$, we have $c\br{u,x} \leq c\br{v,x}$.
    \item \textbf{Arbitrary costs:} No additional restrictions on the cost function are imposed.
\end{itemize}

\subsection{Motivations and Applications}

Fast retrieval of information in graph structures is a well-studied problem, 
beginning with the seminal work of Knuth \cite{Knuth1973}. 
When the underlying search space provides non-local information about the target, 
the search process can be accelerated, since each query rules out a large set of 
possible target locations. The goal of designing search strategies is to find ways of exploiting this property as effectively as possible. 

Searching arises in various practical problems, but it can also be reformulated, to fit a wide range of real-life applications. Searching in graphs can be used to model a variety of problems, that may initially appear unrelated. These include: scheduling of parallel database join operations \cite{DereniowskiEfPQProcByGRank,OnMinERSTs,MinERSTrofTGs}, automated bug detection in computer code \cite{OptimalSinT,dereniowski2022CFApproxAlgForBSInTsWithMonoQTimes,dereniowski2024SInTsMonoQTs,szyfelbein2025searchingtreeskupmodularweight}, parallel Cholesky factorization of matrices \cite{Dereniowski2003CholeskyFactofMx}, VLSI-layouts \cite{OnAGPartWAppVLSI}, hierarchical clustering of data \cite{Acostfunctionforsimilaritybasedhierarchicalclustering,HCObjFsandAlgs,Approximatehierarchicalclusteringviasparsestcutandspreadingmetrics} and parallel assembly of multi-part products from their components \cite{ParAofModPs,Dereniowski2009ERankOfWTs}.

We focus on the average-case version of the problem rather than the worst-case, 
since it is natural to assume that the search strategies we design are intended 
to be used repeatedly. 
This motivates the introduction of a weight function, as some vertices may serve 
as targets more frequently than others. 
We also allow arbitrary query costs, since performing a query may require 
significant resources, such as time or money. 

Studying cost functions with the monotonicity property is motivated by the problem of detecting a hidden source of a spreading pollution in a network. The pollution spreads from the source along the edges of the network in a way which causes a dilution of the pollutant concentration with increasing distance from the source. The searcher is allowed to take samples at various nodes of the network. Each sample reveales a direction towards the source. However, since at larger distances the concentration of the pollutant is lower, taking samples further away from the source is more costly. Additionally, since the amount of the pollutant released into the network is unknown, the searcher cannot predict the distance from the source based on the concentration measured at a given node. The goal of the searcher is to find a strategy of taking samples to find the source of the pollution as efficiently as possible.

\subsection{Related work}\label{relatedwork}

Searching in graphs is a generalization of the classic Binary Search Problem. In the classical setting, the searcher performs comparisons between the target and the queried element, which allows them to determine whether the target is smaller, larger or equal to the queried element. This is equivalent to searching in a path with uniform costs.
The Graph Search Problem and its variants are related to multiple independently studied problems. These include, among others: 
\begin{itemize}
    \item Binary Search \cite{OnakParys2006GenOfBSSInTsAndFLikePosets,dereniowski2017ApproxSsForGeneralBSinWTs,Deligkas2019BsInGsRev,Emamjomeh2016DetAndProbBSinGs,dereniowski2022CFApproxAlgForBSInTsWithMonoQTimes,dereniowski2024SInTsMonoQTs,noisyBSSFC,Dereniowski2024OnMG,EfficientNoisyBinarySearch,Dereniowski2023Edge},
    \item Tree Search Problem \cite{Jacobs2010OnTheComplexSearchInTsAvg,Cicalese2014ImprovedApproxAvgTs,Cicalese2016OnTSPwNonUniCost}, 
    \item Binary Identification Problem \cite{Cicalese2012BinIdentPForWTs}, 
    \item Ranking Colorings \cite{Knuth1973,Dereniowski2009ERankOfWTs,DereniowskiERAndSInPOSets,DereniowskiEfPQProcByGRank,DereniowskiVxRankOfChGsAndWTs,Lam1998ERankOfGsIsH}, 
    \item Ordered Colorings \cite{KATCHALSKI1995141}, 
    \item Elimination Trees \cite{Pothen1988OptimalEliminationTrees}, 
    \item Hub Labeling \cite{Angelidakis2018ShortestPQ},
    \item Tree-Depth \cite{NESETRIL20061022,BOROWIECKI2023113682},
    \item Partition Trees \cite{Hgemo2024TightAB},
    \item Hierarchical Clustering \cite{Acostfunctionforsimilaritybasedhierarchicalclustering,HCObjFsandAlgs,Approximatehierarchicalclusteringviasparsestcutandspreadingmetrics}, 
    \item Search Trees on Trees \cite{SplayTonT,Fast_app_centroid_trees}, 
    \item LIFO-Search \cite{GIANNOPOULOU20122089}. 
\end{itemize}

A variant of the Graph Search Problem can also be formulated in which queries 
are performed on edges rather than vertices. 
All of the above definitions are equivalent when the input graph is a tree (for edge and vertex query model respectively). However, when generalizing to arbitrary graphs, the search process may be defined in various ways. For example, the reply may be the connected component of $G-v$ containing $x$ (the model we are interested in), or the direction towards $x$ (i.e. a neighbor of $v$ on a shortest path from $v$ to $x$). In what follows, we summarize the most important and relevant results, 
organized according to the query model and the objective function\footnote{Note that the notion of vertex weights is relevant only in the average-case setting.}.
\subsubsection{Vertex queries, worst-case cost}
When the input graph is a tree and all costs are uniform, the problem is solvable in linear time \cite{Schaffer1989OptNodeRankOfTsInLinTime,OnakParys2006GenOfBSSInTsAndFLikePosets} and $O\br{\log n}$ queries always suffice. Beyond trees, the problem is known to be NP-hard even in: chordal graphs \cite{DereniowskiVxRankOfChGsAndWTs}, bipartite and co-bipartite graphs \cite{RankingsofGraphs}. By combining the results of \cite{BODLAENDER1995238} and \cite{Improvedapproximationalgorithmsvertexseparators} one can obtain a general $O\br{\log^{3/2}n}$-approximation algorithm. Additionally, the problem is solvable in polynomial time for graphs with bounded treewidth \cite{RankingsofGraphs}. 

For non-uniform costs independent on the target, the problem is known to be NP-hard even on trees 
\cite{Dereniowski2009ERankOfWTs,Cicalese2012BinIdentPForWTs,Cicalese2016OnTSPwNonUniCost}, 
for which there exists an $O\br{\sqrt{\log n}}$-approximation algorithm 
\cite{dereniowski2017ApproxSsForGeneralBSinWTs}. 
Further improvements are possible for restricted classes of the cost function. 
These results include:
\begin{itemize}
    \item A parametrized PTAS running in $O\br{\br{cn/\epsilon}^{2c/\epsilon}}$ time, where $c$ is the largest cost \cite{DereniowskiVxRankOfChGsAndWTs}, 
    \item A 2-approximation algorithm for down-monotonic cost functions \cite{dereniowski2022CFApproxAlgForBSInTsWithMonoQTimes},
    \item An 8-approximation algorithm for up-monotonic cost functions \cite{dereniowski2022CFApproxAlgForBSInTsWithMonoQTimes,dereniowski2024SInTsMonoQTs},
    \item A parametrized $O\br{\log\log n}$-approximation algorithm for $k$-up-modular cost functions, running in $k^{O\br{\log k}}\cdot\text{poly}\br{n}$ time \cite{szyfelbein2025searchingtreeskupmodularweight}.
\end{itemize}

For general graphs with non-uniform costs, the problem can be reduced to the uniform-cost case 
by rounding all costs to polynomial values and applying a reduction similar to that in 
\cite{DereniowskiVxRankOfChGsAndWTs}, which yields an $O\br{\log^{3/2} n}$-approximation.

The distance-dependent cost version for trees has been recently introduced by \cite{BSwDDCosts}. They proved that for paths and the worst case criterion the clasic binary search strategy provides an 4-approximation. For trees and worst case criterion they showed a PTAS running in time $O\br{n^{2+\br{p^2/2+3p/2+2}\cdot\br{1+2/\epsilon}}}$, where $p$ is the degree of the polynomial function $f$ such that for every $u,v \in V\br{T}$, $c\br{u,v}=f\br{d\br{u,v}}$. 

\subsubsection{Vertex queries, average-case cost}

For trees with non-uniform weights and uniform query costs, a PTAS running in 
$O\br{n^{2/\epsilon+1}}$ time exists \cite{SplayTonT}. 
This was later improved to an FPTAS with running time 
$O\br{\br{1/\epsilon}^{2/\log_23}\cdot n^{1+4/\log_23}\cdot \log^2\br{n/\epsilon}}$ 
\cite{Berendsohn2024}. 
Every optimal decision tree has height at most $O\br{\log w\br{T}}$, and a simple 
decision tree that always queries the weighted centroid achieves a 2-approximation which is tight
\cite{Fast_app_centroid_trees}. 
Beyond trees, the problem is NP-hard even for graphs with treewidth at most 15, 
as well as for dense graphs with uniform weights \cite{Berendsohn2024}. 

\subsubsection{Edge queries, worst-case cost}

For trees with uniform costs, the problem is solvable in linear time 
\cite{Lam1998ERankOfGsIsH,Mozes_Onak2008FindOptTSStartInLinTime}, 
and at most 
$\frac{\Delta\br{T}-1}{\log\br{\Delta\br{T}+1}-1}\cdot \log n$ queries always suffice 
\cite{Emamjomeh2016DetAndProbBSinGs}. 
For general graphs, the problem is known to be NP-hard \cite{Lam1998ERankOfGsIsH}. 
When query costs are non-uniform non dependent on target, an $O\br{\log n}$-approximation was given 
\cite{Dereniowski2009ERankOfWTs}, later improved to 
$O\br{\log n/\log\log\log n}$ \cite{Cicalese2012BinIdentPForWTs}, then to 
$O\br{\log n/\log\log n}$, and finally, by reduction from the vertex version of the 
problem, to $O\br{\sqrt{\log n}}$ \cite{dereniowski2017ApproxSsForGeneralBSinWTs}.
\subsubsection{Edge queries, average-case cost}

For trees with uniform costs, the problem is known to be (weakly) NP-hard 
even for trees with diameter at most $4$ and for trees with degree at most $16$ 
\cite{Jacobs2010OnTheComplexSearchInTsAvg}. 
Every optimal decision tree has height at most $O\br{\Delta\br{T}\cdot \log w\br{T}}$, and there exists 
a parametrized FPTAS running in $\text{poly}\br{n^{\Delta\br{T}}/\epsilon}$ time. 
A simple greedy algorithm that always queries the edge which splits the weights most evenly 
can be shown to achieve a 2-approximation \cite{Jacobs2010OnTheComplexSearchInTsAvg}, later improved to $\phi$ in \cite{Cicalese2014ImprovedApproxAvgTs} and $3/2$ in 
\cite{Hgemo2024TightAB} which is tight. 
For general graphs the problem cannot be approximated within any constant factor 
if the {Small Set Expansion Hypothesis} holds 
\cite{Approximatehierarchicalclusteringviasparsestcutandspreadingmetrics}, even when all costs are uniform. For uniform weights and non-uniform costs, there exists an $6.75$-approximation algorithm for trees 
and an $O\br{\sqrt{\log n}}$-approximation algorithm for general graphs 
\cite{HCObjFsandAlgs,Approximatehierarchicalclusteringviasparsestcutandspreadingmetrics}. 

\subsection{Organization of the paper and our results}

In Section \ref{sec:notions}, we introduce all necessary notions, 
preliminaries, and formal definitions required for the analysis, including:
an alternative notion of strategies called query assignment and its two relaxations: separator and pseudo-separator assignments (Section \ref{subsec:query-assignment}). We also introduce the notion of 
vertex cuts and separators (Section \ref{cutsAndSeparators}) and basic lower bounds on the cost of an optimal decision tree (Section \ref{subsec:lower-bounds}).

In Section \ref{serachingInTs}, we consider cost function independent of the target and trees. We start by divising a bicriteria FPTAS for the Weighted 
$\alpha$-Separator Problem. Then, in Section \ref{HowToSearchInTs} we show how to use the later procedure to obtain a 
$\br{4+\epsilon}$-approximation algorithm for the Tree Search Problem running in $O\br{n^4/\epsilon^2}$ time. The algorithm is based on recursive application of the aforementioned FPTAS to find a good separator and building the decision greedily around it. 

In Section \ref{serachingInGs}, we mentain the query cost model but shift our attention to general graphs. 
Using the $O\br{\sqrt{\log n}}$-approximation algorithm for the Min-Ratio 
Vertex Cut Problem from \cite{Improvedapproximationalgorithmsvertexseparators}, 
we obtain a polynomial-time $O\br{\sqrt{\log n}}$-approximation algorithm 
for the Graph Search Problem. The algorithm is similar to the one for trees, but we use the Min-Ratio Vertex Cut Problem instead of the Weighted $\alpha$-Separator Problem to build the decision tree. Additionally, in Section \ref{subsec:graphs-parametrized} we give an FPTAS parametrized by the number of non-leaf vertices of the graph based on a generalization of the dynamic programming algorithm for the Scheduling with Rejection \cite{scheduling_with_rejection}.

In Section \ref{sec:Average-case searching} we change the query cost model to depend also on the target vertex but we require that the costs are monotone, non-decreasing with respect to the target. We present a $2$-approximation for this case. We start with an LP-relaxation of the problem. In Section \ref{subsec:avg-algorithm} we show how to use the solution to this LP to construct a decision tree with cost at most twice the optimal. To do so we require two subroutines: The first one uses the solution to the LP to construct a pseudo-separator assignment which is a relaxed notion of a strategy (Section \ref{subsec:constructing-pseudo-separator}). The second subprocedure is used to reconstruct the decision tree from the pseudo-separator assignment (Section \ref{subsec:reconstructing-decision-tree}). In section \ref{sec:worst-case-searching}, we show how can we slighlty alter the LP to also obtain a 2-approximation for the worst-case version of problem.

In Section \ref{sec:hardness}, we show two hardness results. Firstly, in Section \ref{subsec:hardness-trees-nondep} we consider the case of cost function which depends only on the queried vertex. We prove that the problem is (weakly) NP-hard even when restricted to trees with $\Delta\br{T}\leq 16$ or to trees with $\diam\br{T}\leq 6$. 
Secondly, in Section \ref{subsec:hardness-stars-dep} we consider the case of stars with arbitrary cost function of the queried vertex and the target and we show that the problem cannot be approximated within any constant factor under the Unique Games Conjecture, even when restricted to stars. Importantly, the problem generalizes the well known Minimum Feedback Arc Set Problem. Since the state-of-the-art approximation algorithm for the latter achieves an $O\br{\log n \cdot\log \log n}$-approximation \cite{ApproximatingMinimumFeedbackSetsAndMulticutsInDirectedGraphs}, it seems to be a very challenging task to obtain efficient approximation algorithms for our problem even for very simply graph topologies.


\section{Notions and Preliminaries}\label{sec:notions}
We assume that every graph $G$ considered is simple and connected. 
By $uv \in E\br{G}$ we denote an edge connecting vertices $u$ and $v$ in $G$. 
Let $v \in V\br{G}$. By $G-v$ we denote the set of connected components 
resulting from deleting $v$ from $G$. 
For a set $S \subseteq V\br{G}$, $G-S$ denotes the set of 
connected components resulting from deleting all vertices in $S$ from $G$. 
For a family of subsets $\mathcal{F}$ of $V\br{G}$, we define 
$G-\mathcal{F} = G-\bigcup_{S \in \mathcal{F}} S$. The set of neighbors of a vertex $v \in V\br{G}$ is denoted by 
$N_G\br{v} = \brc{u \in V\br{G} \colon uv \in E\br{G}}$, 
and the set of neighbors of a subgraph $\mathcal{G}$ of $G$ is 
$N_G\br{\mathcal{G}} = \bigcup_{v \in V\br{\mathcal{G}}} N_G\br{v} - V\br{\mathcal{G}}$. For any function $f\colon V\br{G} \to \mathbb{N}$ and any set $S \subseteq V\br{G}$, 
we define $f\br{S} = \sum_{v \in S} f\br{v}$, 
and for a graph $G$, $f\br{G} = f\br{V\br{G}}$.  

Let $T$ be a tree. If $T$ is rooted, its root is denoted by $r\br{T}$. 
For any vertex $v \in V\br{T}$, let 
$\mathcal{C}_{T,v} = \brc{c_1, c_2, \dots, c_{\deg_{T,v}^+}}$ 
be the set of children of $v$. By $T_v$ we will denote the subtree of $T$ rooted at $v$ with all its descendants, and by $T_{v,i}$ we will denote the subtree of $T$ consisting of $v$ and $T_{c_1},\dots, T_{c_i}$.
For a subset $S \subseteq V\br{T}$, $T\angl{S}$ denotes the minimal connected 
subtree of $T$ containing all vertices in $S$. For any $u,v \in V\br{T}$ we denote $P_{u,v}=T\angl{\brc{u,v}}$.
If a given decision tree is optimal, we denote its cost by $\OPT\br{G}$. 
Let $D'$ be a subtree of a valid decision tree $D$ for $T$ containing $r\br{D}$. 
We say that $D'$ is a \textit{partial} decision tree for $D$, 
and we define its cost analogously to that of $D$, 
although we only count the queries belonging to $D'$.

To obtain our approximation algorithms, we will use two different approaches which share a common idea of decomposing the decision tree into a family of separators, which allows us to obtain tight lower bounds on $\OPT\br{G}$. For cost functions non-dependent on the target, we use a greedy approach based on vertex cuts/separators. By showing how to decompose the optimal decision tree into a sequence of separators, we obtain a tight lower bound on $\OPT\br{G}$. For monotone cost functions on trees, we use a twofold relaxation. Firstly, we introduce an alternative notion of strategy, called query assignments. Then we relax this notion into the separator assignments. This will come convenient when formulating the problem as an ILP. At last we present a further relaxation of separator assignments, called pseudo-separator assignments, which will be useful in the rounding step of our approximation algorithms. 

\subsection{Query assignment}
\label{subsec:query-assignment}
In order to find a decision trees for the monotone cost function on trees we focus ourselves on finding an equivalent structure called query sequence assignment. A \textit{query assignment} is a function $Q\colon V\to 2^{V\br{T}}$ such that: For any $v,u\in V\br{T}$ (not necessarily different), $Q\br{v}\cap Q\br{u} \cap P_{v,u}\neq\emptyset$. We also require that for any subtree $T'$ of $T$, $\bigcap_{v\in V\br{T'}} Q\br{v} \neq \emptyset$. To see that the query assignment is an equivalent notion of strategy note that one may build a decision tree $D$ from a query assignment $Q$ as follows: The root of $D$ is $r\in \bigcap_{v\in V\br{T}} Q\br{v}$ and for every connected component $H\in T-r$, the subtree of $D$ corresponding to $H$ is built recursively using the restriction of $Q$ to $H$. Conversely, given a decision tree $D$, one may define a query assignment $Q$ as follows: For every $v\in V\br{T}$, let $Q\br{v}=Q_{T}\br{D,v}$. Since for any vertices $u,v\in V\br{T}$ the strategy queries some vertex $w\in P_{u,v}$, we have that $w \in Q\br{u}$ and $w \in Q\br{v}$ as required and for any subtree $T'$ of $T$, there exists a vertex $r\in V\br{T'}$ which is queried before every other vertex from $V\br{T'}$ in $D$, so $r\in \bigcap_{v\in V\br{T'}} Q\br{v}$ as required.

In order to build our approximation algorithms we will relax the second condition of the above strategy representation. A \textit{separator assignment} is a function $S\colon V\to 2^{V\br{T}}$ such that: For any $v,u\in V\br{T}$, $S\br{v}\cap S\br{u} \cap P_{v,u}\neq\emptyset$ (neither $u$ and $v$ nor $u'$ and $v'$ must be different). The existence of a common vertex for subtrees of $T$ is not required. We will be concerned with finding a good separator assignment minimizing either average-case cost:
$$\sum_{x\in V\br{T}}w\br{x}\cdot\sum_{v\in S\br{x}} c\br{v, x}$$
or worst-case cost:
$$\max_{x\in V\br{T}}\sum_{v\in S\br{v}} c\br{v, x}.$$

We further relax the condition on the separator assignment as follows:
A \textit{pseudo-separator assignment} is a function $S\colon V\to 2^{V\br{T}}$ such that: For any $v,u\in V\br{T}$ (not necessarily different), there exist vertices $v'\in S\br{v}$ and $u'\in S\br{u}$ such that $u'\in P_{v, v'}$. and $v'\in P_{u, u'}$. This notion will come useful while constructing a resulting decision tree during the rounding step of our LP based approximation algorithms.
\subsection{Cuts and separators}\label{cutsAndSeparators}
To obtain a tight lower bound on the cost of our solution for the case when the cost function does not depend on the target, 
we establish a connection between GSP and vertex cuts/separator problems. Our analysis is inspired by the edge query variant of the GSP \cite{Approximatehierarchicalclusteringviasparsestcutandspreadingmetrics}. However, since the query model is different, to obtain our results we use connection to the vertex separation instead of edge cutting.

Any subset $S\subseteq V\br{G}$ is called \textit{$\alpha$-separator} if for every $H\in G-S$ we have $w\br{H}\leq w\br{G}/\alpha$. Any partition $\br{A,S,B}$ of $V\br{G}$ is called a \textit{vertex-cut} if there are no $u\in A$ and $v\in B$ for which $uv\in E\br{G}$.
The \textit{Weighted $\alpha$-Separator Problem} is as follows:

\begin{tcolorbox}[colback=white, title=Weighted $\alpha$-Separator Problem, fonttitle=\bfseries, breakable]
\textbf{Input:} Graph $G=\br{V\br{G}, E\br{G}}$, a cost function $c\colon V\to \mathbb{N}$, a weight function $w\colon V\to \mathbb{N}$ and $\alpha\in\mathbb{R}^+$.

\textbf{Output:} An $\alpha$-separator such that $c\br{S}$ is minimized.
\end{tcolorbox}
We also define the Min-Ratio Vertex Cut Problem:

\begin{tcolorbox}[colback=white, title= Min-Ratio Vertex Cut Problem, fonttitle=\bfseries, breakable]
\textbf{Input:} Graph $G=\br{V\br{G}, E\br{G}}$, the cost function $c\colon V\to \mathbb{N}$ and the weight function $w\colon V\to \mathbb{N}$.

\textbf{Output:} A vertex cut $\br{A,S,B}$ of $V\br{G}$ minimizing the ratio:
$$
\alpha_{c,w}\br{A,S,B}=\frac{c\br{S}}{w\br{A\cup S}\cdot w\br{B\cup S}}.
$$
\end{tcolorbox}
\subsection{Levels of $\OPT$ and basic bounds}\label{subsec:lower-bounds}

We start by introducing the following reinterpretation of the cost function. 
For each node $v \in D$, let $G_{D,v}$ be the subgraph of $G$ in which 
$v$ is queried when using $D$. For any graph $G$ and decision tree $D$, denote by 
$\mathcal{R}_D\br{G} = \brc{V\br{G_{D,v}}\colon v \in V\br{G}}$ 
the family of all candidate subsets which can occur while searching according to $D$ in $G$. 

Let $D^*$ be an arbitrary decision tree for the Graph Search Problem such that 
$c_G\br{D^*} = \OPT\br{G}$. 
We denote by $\mathcal{L}_{k}^*$ the subfamily of $\mathcal{R}_{D^*}\br{G}$ 
consisting of all maximal elements $H$ of $\mathcal{R}_{D^*}\br{G}$ with $w\br{H} \leq k$, 
that is, if some superset $H'$ of $H$ belongs to $\mathcal{R}_{D^*}\br{G}$, 
then $w\br{H'} > k$. 
We call such a set the $k$-th \textit{level} of $\OPT\br{G}$. 
Let $S_{k}^* = V\br{G} - \mathcal{L}_{k}^*$. 
These are the vertices that we describe as belonging to the separator at the $k$-th level of $D^*$.

Notice that $S_{k}^*$ forms a Weighted $w\br{G}/k$-separator of $G$. 
Furthermore, for any $H_1, H_2 \in \mathcal{R}_D\br{G}$, we have 
$H_1 \cup H_2 \neq \emptyset$ if and only if $H_1 \subseteq H_2$ or 
$H_2 \subseteq H_1$, so $\mathcal{R}_D\br{G}$ is laminar. 
Therefore, for any $k_1 \neq k_2$, we have 
$\mathcal{L}_{k_1}^* \cap \mathcal{L}_{k_2}^* = \emptyset$.

For the rest of this section assume that the cost function does not depend on the target, i.e., for every $v,x \in V\br{G}$, $c\br{v,x} = c\br{v}$.
We observe that for any $v\in V\br{G}$ the query cost of $v$ contributes to the search cost of every $u\in V\br{G_{D,v}}$ so the total contribution of $v$ to the cost is $w\br{G_{D,v}} \cdot c\br{v}$, 
and therefore it is easy to see that:

\begin{lemma}\label{contributionLemma}
$$
c_G\br{D} = \sum_{v \in V\br{G}} w\br{G_{D,v}} \cdot c\br{v}.
$$  
\end{lemma}

Using the latter we get the following lemma, which allows us to decompose the value of $\OPT\br{G}$ into cost of the separators at each level:
\begin{lemma}
                $$\OPT\br{G}=\sum_{k=0}^{w\br{G}-1}c\br{S_{k}^*}.$$
            \end{lemma}
            \begin{proof}
                Consider any vertex $v$. For every $0\leq k<w\br{G_{D^*,v}}$, $v\notin \bigcup_{H\in \mathcal{L}_{k}^*}H$, so $v\in S_{k}^*$ and the contribution of $v$ to the cost is $w\br{G_{D^*,v}}\cdot c\br{v}$:
                $$\sum_{k=0}^{w\br{G}-1}c\br{S_{k}^*}=\sum_{v\in V\br{G}}\sum_{k=0}^{w\br{G_{D^*,v}}-1}c\br{v}=\OPT\br{G}$$
                
                where the second equality is by Lemma \ref{contributionLemma}.
            \end{proof}
            
Using the above lemma one easily obtains the following lower bound on the cost of the optimal solution:
\begin{lemma}\label{lb_opt}
            $$
            2\cdot\OPT\br{G}= 2\cdot\sum_{k=0}^{w\br{T}-1}c\br{S_{k}^*} \geq \sum_{k=0}^{w\br{T}}c\br{S_{\fl{k/2}}^*}.
            $$
\end{lemma}

Intuitively, the above reindexing allows us to lower bound the value of $\OPT\br{G}$ by the sum of minimal costs required to separate each candidate subgraph at some level $k$ into candidate subgraphs on levels at most $\fl{k/2}$.

In order to upper bound the cost of decision trees produced by our algorithms we will use the following lemma:
\begin{lemma}\label{splitting}
    Let $\mathcal{G}$ be any subgraph of $G$ and $0\leq\beta\leq 1$. Then: 
            $$
           \beta\cdot w\br{\mathcal{G}}\cdot c\br{S_{\fl{w\br{\mathcal{G}}/2}}^*\cap \mathcal{G}}
            \leq \sum_{k=\br{1-\beta}w\br{\mathcal{G}}+1}^{w\br{\mathcal{G}}}c\br{S_{\fl{k/2}}^*\cap \mathcal{G}}.
            $$
\end{lemma}

    \begin{proof}
        The inequality is due to the fact that as $k$ decreases, more vertices belong to the separator. 
    \end{proof}
    
The following lemma qill be required in order to attach decision trees 
below already constructed partial decision trees without ambiguity:
\begin{lemma}\label{neighborsPathLemma}
     Let $\mathcal{G}$ be a connected subgraph of a graph $G$ and $D$ be a partial decision tree for $G$ having no queries to vertices of $\mathcal{G}$, but containing at least one query to $N_{G}\br{V\br{\mathcal{G}}}$. Additionally, let $Q$ be the set of all queries to vertices from $N_{G}\br{V\br{\mathcal{G}}}$ in $D$. Then $D\angl{Q}$ forms a path in $D$. 
\end{lemma}
    \begin{proof}
        Let $q$ be any query in $D$. There are two cases:
        \begin{enumerate}
            \item $q\in V\br{G-V\br{\mathcal{G}}-N_{G}\br{V\br{\mathcal{G}}}}$. Then, for every $x\in N_{G}\br{V\br{\mathcal{G}}}$ being the target, the answer is the same connected component of $G-q$. Therefore, $q$ has at most one child $u$ in $D$, such that $V\br{D_u}\cap Q \neq \emptyset$.
            \item $q\in N_{G}\br{V\br{\mathcal{G}}}$. After a query to $q$, the situation is the same as in the first case, except when $x=q$. Then the response is $x$ itself, in which case no further queries are needed, and again $q$ has at most one child $u$ in $D$, such that $V\br{D_u}\cap Q \neq \emptyset$.
        \end{enumerate}
    \end{proof}
    
This lemma will become useful in the following scenario: 
Let $D$ be a partial decision tree for $G$ satisfying the conditions of the lemma and $D_{\mathcal{G}}$ be any partial decision tree for a subgraph $\mathcal{G}$ of $G$. 
Since all of the queries in $Q$ form a path, without ambiguity, we can attach $D_{\mathcal{G}}$ to $D$ 
below the last query of $Q$, thereby obtaining a valid partial decision tree for $G$.

\section{Trees, cost function nondependent on the target}\label{serachingInTs}

In this section, we present a $\br{4+\epsilon}$-approximation FPTAS 
for the case where the input graph is a tree and the cost function is nondependent on the target.
To achieve this result, we use the connection between searching in trees and 
the Weighted $\alpha$-Separator Problem. 
This connection provides a lower-bounding scheme for our recursive algorithm, 
which at each level of recursion, constructs a decision tree using the 
$\alpha$-separator obtained by the following procedure:

\begin{theorem}\label{bicriteriaFPTAS}
    Let $S$ be an optimal weighted $\alpha$-separator for $\br{T,c,w,\alpha}$. For any $\delta>0$ there exists an algorithm \FSeparatorFPTAS, which returns a separator $S'$, such that:
    \begin{enumerate}
        \item $c\br{S'}\leq c\br{S}$.
        \item $w\br{H}\leq \frac{\br{1+\delta}\cdot w\br{T}}{\alpha}$ for every $H\in T-S'$.
        \item The algorithm runs in $O\br{n^3/\delta^2}$ time.
    \end{enumerate}
\end{theorem}
\begin{proof}
    
We devise a dynamic programming procedure similar to the one in \cite{kseparator} and combine it with a rounding trick to obtain a bi-criteria FPTAS.
Note that the authors considered only the case in which all weights are uniform. We generalize their algorithm to arbitrary integer weights and introduce an additional case that 
was previously lacking\footnote{Probably due to an oversight.}.

\begin{theorem}\label{separator}
    Let $T$ be a tree. 
    There exists a pseudopolynomial time algorithm for the Weighted $\alpha$-Separator Problem running in 
    $O\br{n\cdot \br{w\br{T}/\alpha}^2}$ time.
\end{theorem}

\begin{proof}
    Assume that the input tree is rooted at an arbitrary vertex $r\br{T}$. Let $k=\fl{w\br{T}/\alpha}$. We want to find a separator $S$ such that for every $H\in T-S$, $w\br{H}\leq k$.
    Let $C_{v}$ denote the cost of the optimal separator $S_v$ in $T_v$ with this property. 
    Define $C_{v}^{in}$ as the cost of the optimal separator for $T_v$, 
    under the condition that $v \in S_v$. 
    We immediately have:
    $$
    C_{v}^{in} = c\br{v}+\sum_{c\in \mathcal{C}_{T,v}}C_{c}.
    $$
    
    Assume that $v \notin S_v$. 
    Let $H_v \in T_v - S_v$ be the component containing $v$. 
    For every integer $0 \leq w \leq k$, let $C_v^{out}[w]$ be the cost of the optimal separator for $T_v$, 
    such that $v \notin S_v$ and $w\br{H_v} = w$. 
    Then:
    $$
    C_v = \min \brc{C_{v}^{in}, \min_{0 \leq w \leq k} C_v^{out}[w]}.
    $$
    
    For any vertex $v \in V\br{T}$ and any integer $1 \leq i \leq \deg_{T,v}^+$, 
    let $S_{v,i}$ be the optimal separator for $T_{v,i}$ and $H_{v,i} \in T_{v,i} - S_{v,i}$ be the component containing $v$.
    For any integer $0 \leq w \leq k$, let $C_{v,i}^{out}[w]$ be the cost of an optimal separator for $T_{v,i}$, 
    such that $v \notin S_{v,i}$ and $w\br{H_{v,i}} = w$. 
    Then
    $$
    C_{v}^{out}[w] = C_{v,\deg_{T,v}^+}^{out}[w].
    $$
    
    For $i = 1$ we have:
    $$
    C_{v,1}^{out}[w] =
    \begin{cases}
        \infty, & \text{if } w < w\br{v},\\
        \min \brc{C_{c_1}^{in},C_{c_1}^{out}[0]}, & \text{if } w = w\br{v},\\
        C_{c_1}^{out}[w - w\br{v}], & \text{if } w > w\br{v}.
    \end{cases}
    $$
    
    For $i > 1$:
    $$
    C_{v,i}^{out}[w] = 
    \min \brc{
        C_{v,i-1}^{out}[w] + C_{c_i}^{in}, 
        \min_{0 \leq j \leq w} \brc{ C_{v,i-1}^{out}[w-j] + C_{c_i}^{out}[j] }
    }.
    $$
    
    In the above, the first term of the outer minimum corresponds to the case $c_i \in S_{v,i}$, 
    so $H_{v,i} = H_{v,i-1}$. 
    The second term considers the alternative, checking all possible partitions of 
    the weight between $H_{v,i-1}$ and $H_{c_i}$.

    The above recurrence relationships suffice to compute $C_{r\br{T}}$, the cost of the optimal separator $S$ for $T$. 
    Computation is performed in a bottom-up, left-to-right manner, starting from the leaves. 
    For a leaf $v$, we have $C_v^{in} = c\br{v}$ and:
    $$C_v^{out}[w] = \begin{cases}
        0, & \text{if } w = w\br{v}\leq k,\\
        \infty, & \text{otherwise.}
    \end{cases}$$ 

    Since each of the $C_v^{in}$ subproblems requires $O\br{\deg_{T,v}^+}$ computational steps we get that overall, they require $O\br{n}$ running time. As there are $O\br{n\cdot k} = O\br{n \cdot w\br{T}/\alpha}$ remaining subproblems and each requires 
    at most $O\br{k}=O\br{w\br{T}/\alpha}$ computational steps, the running time is $O\br{n \cdot \br{w\br{T}/\alpha}^2}$.
\end{proof}

The running time of the above procedure depends on $w\br{T}$ which may not be polynomial. 
To alleviate this difficulty, we slightly relax the condition on the size of components 
in $T-S$ using a controlled parameter $\delta$. 
Based on this relaxation, we show how to construct a bicriteria FPTAS for the problem. Let $\delta>0$ be any fixed constant and let $\FSeparator$ be the dynamic programming procedure from Theorem \ref{separator}.
        The algorithm is as follows:
        
\begin{algorithm}[H]
\caption{The bicriteria FPTAS for the Weighted $\alpha$-separator Problem}
\label{alg:SeparatorFPTAS}
\SetKwFunction{FDecisionTree}{DecisionTree}
\SetKwProg{Fn}{Procedure}{:}{}
\Fn{$\FSeparatorFPTAS\br{T, c, w, \alpha, \delta}$}{
$K\gets\frac{\delta\cdot w\br{T}}{n\cdot \alpha}$.

\ForEach{ $v\in V\br{T}$}
{$w'\br{v} \gets \fl{\frac{w\br{v}}{K}}$.}

$\alpha'\gets\frac{\alpha\cdot K\cdot w'\br{T}}{w\br
T}$.

$S'\gets\FSeparator\br{T, c, w', \alpha'}$.

\Return $S'$.
}
\end{algorithm}
        \begin{lemma}
            Let $S$ be the optimal separator for the $\br{T, c, w, \alpha}$ instance. We have that $c\br{S'}\leq c\br{S}$.
        \end{lemma}
            \begin{proof}
                We prove that $S$ is a valid separator for the $\br{T, c, w', \alpha'}$ instance, so that $c\br{S'}\leq c\br{S}$.
                To simplify the analysis, we will define the auxiliary instance: For every $v\in V\br
                {T}$, let $w''\br{v} = K\cdot\fl{\frac{w\br{v}}{K}} $. Additionally, let $ \alpha'' =\frac{\alpha \cdot w''\br{T}}{w\br{T}}$.
                In this new instance, for $v\in V\br{T}$ we have $w''\br{v}\leq w\br{v}$, so for every $H\in T-S$, 
                $$w''\br{H}\leq w\br{H}\leq w\br{T}/\alpha= w''\br{T}/\alpha''$$
                
                where the second inequality is by the definition of the $\alpha$-separator and the equality is by the definition of $\alpha''$.
                
                We conclude that $S$ is an $\alpha''$-separator for the auxiliary instance $\br{T, c, w'', \alpha''}$. Now notice that the $\br{T, c, w', \alpha'}$ instance has all of its weights scaled by a constant value of $K$, relatively to $\br{T, c, w'', \alpha''}$ and $\alpha' = \alpha''$. As multiplying weights by a constant does not influence the validity of a solution, $S$ is an $\alpha'$-separator for $\br{T, w', c, \alpha'}$ and the claim follows.
            \end{proof}
        \begin{lemma}
            For every $H\in T-S'$, we have that $w\br{H}\leq\frac{\br{1+\delta}\cdot w\br{T}}{\alpha}$.
        \end{lemma}
        
            \begin{proof}
                By definition $ \frac{w\br{v}}{K}\leq w'\br{v}+1$ and therefore, also $w\br{v}\leq K\cdot w'\br{v}+K$. We have:
                \begin{align*}
                \sum_{v\in H}w\br{v}\leq K\cdot\sum_{v\in H}w'\br{v}+K\cdot n
                \leq \frac{K\cdot w'\br{T}}{\alpha'}+K\cdot n 
                = \frac{w\br{T}}{\alpha} + \frac{\delta \cdot w\br{T}}{\alpha}=\frac{\br{1+\delta}\cdot w\br{T}}{\alpha}
                \end{align*}
                
                where the second inequality is due to the fact that $S'$ is a $\alpha'$-separator for $\br{T, c, w', \alpha'}$ instance and the first equality is by the definition of $\alpha'$ and $K$.
            \end{proof}
            
        Combining the two above lemmas with the fact that $\frac{w'\br{T}}{\alpha'}=\frac{w\br{T}}{K\cdot \alpha}=n/\delta$ we have that the algorithm runs in time $O\br{n^3/\delta^2}$ as required.
\end{proof}

\subsection{How to search in trees}\label{HowToSearchInTs}
Below, we show how to use the $\FSeparatorFPTAS$ procedure to construct a solution 
for the Tree Search Problem. 
At each level of the recursion, the algorithm greedily finds an (almost) optimal 
weighted $\alpha$-separator of $T$, denoted $S_T$, and then builds an arbitrary 
decision tree $D_T$ using the vertices in $S_T$ (which can be done in $O\br{n^2}$ time). 

Next, for each $H \in T-S_T$, the procedure is called recursively, and each resulting 
decision tree $D_H$ is attached below the appropriate query in $D_T$. The resulting decision tree is then returned by the procedure.
    \begin{theorem}
        For any $\epsilon>0$, there exists $\br{4+\epsilon}$-approximation algorithm for the Tree Search Problem running in time $O\br{n^4/\epsilon^2}$.
    \end{theorem}
        \begin{proof}
            The procedure is as follows: 
            
\begin{algorithm}[H]
\caption{The $\br{4+\epsilon}$-approximation algorithm for the Tree Search Problem}
\label{alg:DecisionTreeTrees}
\SetKwFunction{FDecisionTree}{DecisionTree}
\SetKwFunction{FSeparatorFPTAS}{SeparatorFPTAS}
\SetKwProg{Fn}{Procedure}{:}{}
\Fn{$\FDecisionTree\br{T, c, w,  \epsilon}$}{
$S_T\gets\FSeparatorFPTAS\br{T, c, w, \alpha=2, \delta = \frac{\epsilon}{4+\epsilon}}$.

$D_T\gets$ arbitrary partial decision tree for $T$, built from vertices of $S_T$.

    \ForEach{$H\in T-S_T$}
    {
        $D_H\gets \FDecisionTree\br{H, c, w, \epsilon}$.

        Hang $D_H$ in $D_T$ below the last query to $v\in N_T\br{H}$.
    }   
    \Return $D_T$.
    
}
\end{algorithm}
            \noindent\begin{minipage}{\linewidth}
    \centering
    \begin{minipage}{0.48\textwidth}
        \centering
        \begin{tikzpicture}[scale=0.7]
    \draw[thick, fill=white, drop shadow]
  (0,0) 
  .. controls (1,0) and (1,-6) .. (0,-6)  
  .. controls (-1,-6) and (-1,0) .. (0,0); 
  
  \node at (-0.75, 0)  {$S_T$};

Dots inside
\foreach \y in {-1,-2,-3,-5} {
  \fill (0,\y) circle (4pt);
}
\node at (0, -3.9) {$\vdots$};

\draw[thick, fill=white, drop shadow]
(3,0.75) 
  .. controls (4.5,0.75) and (4.5,-0.75) .. (3,-0.75)  
  .. controls (1.5,-0.75) and (1.5,0.75) .. (3,0.75); 

\node at (3, -0.05)  {$H_1$};

\draw[thick] (0.45,-0.5) -- (1.87,0);

\draw[thick] (0.745,-2.5) -- (2.04,-0.4);

\draw[thick, fill=white, drop shadow]
(3.5,-1) 
  .. controls (4.5,-1) and (4.5,-2) .. (3.5,-2)  
  .. controls (2.5,-2) and (2.5,-1) .. (3.5,-1); 

\node at (3.5, -1.55)  {$H_2$};

\draw[thick] (0.66,-1.6) -- (2.74,-1.5);

\draw[thick, fill=white, drop shadow]
(3.25,-2.25) 
  .. controls (4.5,-2.25) and (4.5,-3.5) .. (3.25,-3.5)  
  .. controls (2,-3.5) and (2,-2.25) .. (3.25,-2.25); 

\node at (3.25, -2.9)  {$H_3$};

\draw[thick] (0.75,-2.6) -- (2.32,-2.8);

\draw[thick] (0.75,-3.5) -- (2.32,-3);

\draw[thick] (0.58,-5) -- (2.6,-3.35);

\node at (3.15, -3.95) {$\vdots$};

\draw[thick, fill=white, drop shadow]
(3,-4.65) 
  .. controls (4.75,-4.65) and (4.75,-6.5) .. (3,-6.5)  
  .. controls (1.25,-6.5) and (1.25,-4.65) .. (3,-4.65); 
  
\node at (3, -5.6)  {$H_p$};

\draw[thick] (0.69,-4.25) -- (1.85,-5.1);

\draw[thick] (0.45,-5.5) -- (1.7,-5.75);

\end{tikzpicture}
    \end{minipage}
\begin{minipage}
    {0.48\textwidth}
    \centering
    \begin{tikzpicture}[scale=1]
            \draw[thick, fill=gray!30, drop shadow] (4,-4) -- (4.9,-5.8) -- (3.1,-5.8) -- cycle
                  node[right] {$D_{T}$};
                  
            \draw[thick, fill=white, drop shadow] (2.5,-6.5) -- (3.2,-7.9) -- (1.8,-7.9) -- cycle
            node[right] {$D_{H_1}$};

            \draw[thick, fill=white, drop shadow] (3.8,-6.4) -- (4.3,-7.4) -- (3.3,-7.4) -- cycle
            node[right] {$D_{H_2}$};
            
            \draw[thick, fill=white, drop shadow] (4.8,-6.3) -- (5.2,-7.1) -- (4.4,-7.1) -- cycle
            node[right] {$D_{H_3}$};
            
            \draw[thick, fill=white, drop shadow] (6.5,-6.5) -- (7.2,-8.1) -- (5.8,-8.1) -- cycle
            node[right] {$D_{H_p}$};
            
            \node at (5.75, -7) {$\dots$};
            
            \draw[thick] (3.1,-5.8) -- (2.5,-6.5);
            \draw[thick] (3.6,-5.8) -- (3.8,-6.4);
            \draw[thick] (4.2,-5.8) -- (4.8,-6.3);
            
            \draw[thick] (4.9,-5.8) -- (6.5,-6.5);
            
        \end{tikzpicture}
\end{minipage}
    
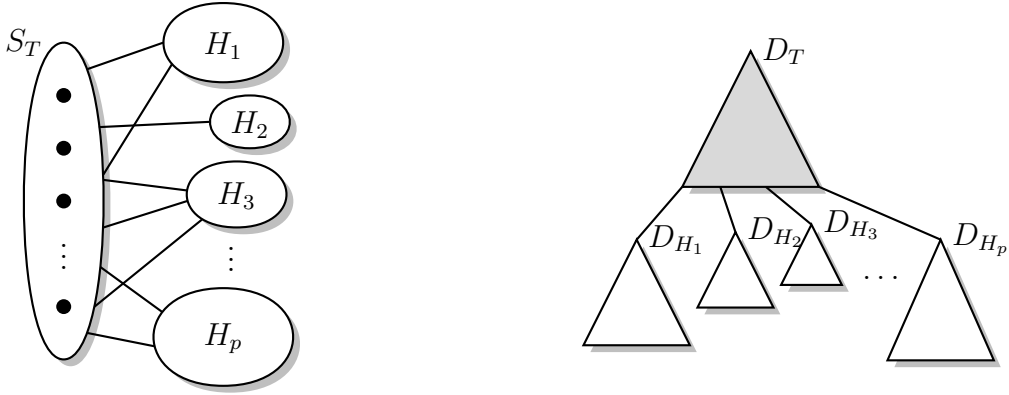
\captionof{figure}{The separator $S_T$ produced by the algorithm and the structure of the decision tree built using $S_T$.}
    \label{fig:placeholder}
\end{minipage}
\par
        
            Let $\mathcal{T}$ be a subtree of $T$ for which the procedure was called and let $S_{\mathcal{T}}^*=S_{\fl{w\br{\mathcal{T}}/2}}^*\cap\mathcal{T}$. By Theorem \ref{bicriteriaFPTAS}, we have that $c\br{S_{\mathcal{T}}}\leq c\br{S_{\mathcal{T}}^*}$. Using $\beta=\frac{1-\delta}{2}$ and applying Lemma \ref{splitting} we have that the contribution of the decision tree $D_{\mathcal{T}}$ is bounded by:
            \begin{align*}
                w\br{\mathcal{T}}\cdot c\br{S_{\mathcal{T}}}
                \leq w\br{\mathcal{T}}\cdot c\br{S_{\mathcal{T}}^*}\leq \frac{2}{1-\delta}\cdot \sum_{k=\frac{1+\delta}{2}\cdot w\br{\mathcal{T}}+1}^{w\br{\mathcal{T}}}c\br{S_{\fl{k/2}}^*\cap \mathcal{T}}.
            \end{align*}
            
            To bound the cost of the whole solution we will firstly show the following lemma:
            \begin{lemma}\label{up_trees}
            $$\sum_{\mathcal{T}}\sum_{k=\frac{1+\delta}{2}\cdot w\br{\mathcal{T}}+1}^{w\br{\mathcal{T}}}c\br{S_{\fl{k/2}}^*\cap \mathcal{T}}\leq \sum_{k=0}^{w\br{T}}c\br{S_{\fl{k/2}}^*}.$$
            \end{lemma}
            \begin{proof}
                Fix a value of $\mathcal{T}$ and $k$. Their contribution to the cost is $c\br{S_{\fl{k/2}}^*\cap \mathcal{T}}$. Consider which candidate subtrees contribute such a term. As $S_{\mathcal{T}}$ is a weighted $\frac{2}{1+\delta}$-separator, we have that $\mathcal{T}$ is the minimal candidate subtree, such that $w\br{\mathcal{T}}\geq k\geq \frac{\br{1+\delta}\cdot w\br{\mathcal{T}}}{2}+1 >w\br{H}$, for every $H\in \mathcal{T}-S_{\mathcal{T}}$. This means that if for every $H\in \mathcal{T}-S_{\mathcal{T}}$, $w\br{H}<k$, then $\mathcal{T}$ contributes such a term. Since for all $H_1, H_2\in \mathcal{T}-S_{\mathcal{T}}$ we have that $H_1\cap H_2=\emptyset$, $\br{S_{\fl{k/2}}^*\cap H_1}\cup \br{S_{\fl{k/2}}^*\cap H_2}=\emptyset$, the claim follows by summing over all values of $k$.
            \end{proof}
            
        We are now ready to bound the cost of the solution. Let $D$ be the decision tree returned by the procedure. Using the fact that by definition $\frac{4}{1-\delta}=4+\epsilon$, we have:
        \begin{align*}
            c_T\br{D} &\leq \sum_{\mathcal{T}} w\br{\mathcal{T}}\cdot c\br{S_\mathcal{T}}
            \leq \frac{2}{1-\delta}\cdot\sum_{\mathcal{T}}\sum_{k=\frac{1+\delta}{2}\cdot w\br{\mathcal{T}}+1}^{w\br{\mathcal{T}}}c\br{S_{\fl{k/2}}^*\cap \mathcal{T}}
            \\&\leq \frac{2}{1-\delta}\cdot\sum_{k=0}^{w\br{T}}c\br{S_{\fl{k/2}}^*}
            \leq \frac{4}{1-\delta}\cdot\OPT\br{T} = \br{4+\epsilon}\cdot\OPT\br{T}
    \end{align*}

    where the third inequality is due to Lemma \ref{up_trees} and the last inequality is by Lemma \ref{lb_opt}.

    As $1/\delta=\frac{4+\epsilon}{\epsilon}=1+4/\epsilon$ and each $v\in V\br T$ belongs to the set $S_{\mathcal{T}}$ exactly once, we have that the overall running time is at most $O\br{n^4/\epsilon^2}$ as required.
        \end{proof}
\section{Graphs, cost function nondependent on the target}\label{serachingInGs}
In this section we shift out attention towards general graphs. We present an $O\br{\sqrt{\log n}}$-approximation algorithm as well as an FPTAS parametrized by the number of non-leaf vertices of the input graph.

\subsection{$O\br{\sqrt{\log n}}$-approximation algorithm}

To obtain the first result we exploit a connection to a different problem, 
namely the Min-Ratio Vertex Cut Problem. 
Let $\alpha_{c,w}\br{G}$ denote the optimal value of such a vertex cut. 
We invoke the following result:

\begin{theorem}[\cite{Improvedapproximationalgorithmsvertexseparators}]\label{approxmrvc}
    Given a graph $G=\br{V\br{G}, E\br{G}}$, the cost function $c:V\to\mathbb{N}$ and the weight function $w:V\to \mathbb{N}$, there exists a
polynomial-time algorithm, which computes a partition $(A, S, B)$, such that:
$$
\alpha_{c,w}\br{A,S,B}=O\br{\sqrt{\log n
}}\cdot\alpha_{c,w}\br{G}.
$$
\end{theorem}

Combining the latter procedure, with the following result, yields a general $O\br{\sqrt{\log n
}}$-approximation algorithm for the Graph Search Problem:
\begin{theorem}
    Let $f_n$ be the approximation ratio of any polynomial time algorithm for the Min-Ratio Vertex Cut Problem. Then, there exists an $O\br{f_n}$-approximation algorithm for the GSP, running in polynomial time.
\end{theorem}
    \begin{proof}
        Let $\FAlgorithmMinCut$ be the procedure which achieves the $f_n$-approximation ratio for the Min-Ratio Vertex Cut Problem.
        The algorithm is as follows:
        
\begin{algorithm}[H]
\caption{The $f_n$-approximation algorithm for the Graph Search Problem.}
\label{alg:DecisionTreeGraphs}
\SetKwFunction{FDecisionTree}{DecisionTree}
\SetKwProg{Fn}{Procedure}{:}{}
\Fn{$\FDecisionTree\br{G,c,w}$}{
$A_G,S_G, B_G\gets\FAlgorithmMinCut\br{G, c, w}$.

$D_G\gets$ arbitrary partial decision tree for $G$, built from vertices of $S_G$.

    \ForEach{$H\in G-S_G$}
    {
        $D_H\gets \FDecisionTree\br{H, c, w}$.

        Hang $D_H$ in $D_G$ below the last query to $v\in N_G\br{H}$.
    }   
    \Return $D_G$.
    
}
\end{algorithm}
        
        Let $\mathcal{G}$ be any subgraph of $G$, for which the procedure was called and let $S_{\mathcal{G}}^*=S_{\fl{w\br{\mathcal{G}}/2}}^*\cap \mathcal{G}$. The next lemma, shows that we can use $S_{\mathcal{G}}^*$ and $\mathcal{G}-S_{\mathcal{G}}^*$ to build a vertex cut of $\mathcal{G}$, thus providing us with an upper bound on the cost of the solution built by the algorithm:
        \begin{lemma}\label{lambda_lemma}
            Let $\mathcal{H}=\mathcal{G}-S_{\mathcal{G}}^*$ and let $\lambda=6+2\sqrt{5}$ be the unique, positive solution of the equation $\frac{1}{4}-\frac{1}{2\sqrt{\lambda}}=\frac{1}{\lambda}$. Then, we can partition $\mathcal{H}$ into two sets, $\mathcal{A}$ and $\mathcal{B}$ such that for $A=\bigcup_{H\in\mathcal{A}}V\br{H}$ and $B=\bigcup_{H\in\mathcal{B}}V\br{H}$, we have:
            $$w\br{A\cup S_{\mathcal{G}}^*}\cdot w\br{B\cup S_{\mathcal{G}}^*}\geq w\br{\mathcal{G}}^2/\lambda.$$
        \end{lemma}
\begin{proof}
                There are two cases:
                \begin{enumerate}
                    \item $w\br{ S_{\mathcal{G}}^*}\geq w\br{\mathcal{G}}/\sqrt{\lambda}$. In this case we take arbitrary partition $\mathcal{A}, \mathcal{B}$ of $\mathcal{H}$. We have:
                    $$w\br{A\cup S_{\mathcal{G}}^*}\cdot w\br{B\cup S_{\mathcal{G}}^*}\geq w\br{ S_{\mathcal{G}}^*}^2 \geq w\br{\mathcal{G}}^2/\lambda.$$
                \item $w\br{ S_{\mathcal{G}}^*} \leq w\br{\mathcal{G}}/\sqrt{\lambda}$.
                For any choice of the partition $\mathcal{A},\mathcal{B}$ of $\mathcal{H}$, we have $\frac{w\br{A\cup B}}{w\br{\mathcal{G}}}\geq 1-\frac{1}{\sqrt{\lambda}}$. We pick $\mathcal{A},\mathcal{B}$ to be a partition of $\mathcal{H}$, such that $w\br{A}\geq w\br{B}\geq \br{\frac{1}{2}-\frac{1}{\sqrt{\lambda}}}\cdot w\br{\mathcal{G}}$ (this is always possible as $\frac{1}{2}-\frac{1}{\sqrt{\lambda}}>0$ and for each $H\in\mathcal{H}$, $w\br{H}\leq w\br{\mathcal{G}}/2$). We have:
                \begin{align*}
                w\br{A\cup S_{\mathcal{G}}^*}\cdot w\br{B\cup S_{\mathcal{G}}^*}&\geq w\br{A}\cdot w\br{B} \geq \br{\br{1-{1}/{\sqrt{\lambda}}}\cdot w\br{\mathcal{G}}-w\br{B}}\cdot w\br{B}\\&\geq  {w\br{\mathcal{G}}^2}/{2}\cdot \br{{1}/{2}-{1}/{\sqrt{\lambda}}} = {w\br{\mathcal{G}}^2}/{\lambda}    
                \end{align*}

                where the third inequality follows since the function $f\br{w\br{B}}=w\br{B}\cdot\br{\br{1-\frac{1}{\sqrt{\lambda}}}\cdot w\br{\mathcal{G}}-w\br{B}}$ is concave and reaches its minimum in the interval 
                $\left[\br{{1}/{2}-{1}/{\sqrt{\lambda}}}\cdot w\br{\mathcal{G}},w\br{\mathcal{G}}/2\right]$
                when $w\br{B}=\br{\frac{1}{2}-\frac{1}{\sqrt{\lambda}}}\cdot w\br{\mathcal{G}}$.
                \end{enumerate}
            \end{proof}
        
        Using the fact, that the partition $\br{A,S_{\mathcal{G}}^*,B}$ in the above lemma is a vertex cut of $\mathcal{G}$, we have the following upper bound on the optimal value of the min-ratio-vertex cut of $\mathcal{G}$,  $\alpha_{c,w}\br{\mathcal{G}}$:
                $$\alpha_{c,w}\br{\mathcal{G}}\leq \alpha_{c,w}\br{A,S_{\mathcal{G}}^*,B}=\frac{c\br{S_{\mathcal{G}}^*}}{w\br{A\cup S_{\mathcal{G}}^*}\cdot w\br{B\cup S_{\mathcal{G}}^*}}\leq \frac{\lambda\cdot c\br{S_{\mathcal{G}}^*}}{w\br{\mathcal{G}}^2}.
                $$
                
        Let $\br{A_{\mathcal{G}},S_{\mathcal{G}}, B_{\mathcal{G}}}$, be the partition returned by $\FAlgorithmMinCut \br{\mathcal{G},c,w}$. Using Theorem \ref{approxmrvc}, we get that:
        $$
        \alpha_{c,w}\br{A_{\mathcal{G}},S_{\mathcal{G}}, B_{\mathcal{G}}}=\frac{c\br{S_{\mathcal{G}}}}{w\br{A_{\mathcal{G}}\cup S_{\mathcal{G}}}\cdot w\br{B_{\mathcal{G}}\cup S_{\mathcal{G}}}}\leq f_n \cdot \frac{\lambda\cdot c\br{S_{\mathcal{G}}^*}} {w\br{\mathcal{G}}^2}.
        $$
        
          Without loss of generality we may assume that $w\br{A_{\mathcal{G}}}\geq w\br{B_{\mathcal{G}}}$. Let $\beta=w\br{B_{\mathcal{G}}\cup S_{\mathcal{G}}}/w\br{\mathcal{G}}$. We have that $\br{1-\beta}\cdot w\br{\mathcal{G}} = w\br{A_\mathcal{G}}$, so we conclude that the contribution of the decision tree $D_{\mathcal{G}}$ is bounded by:
        \begin{align*}
        w\br{\mathcal{G}}\cdot c\br{S_{\mathcal{G}}} &\leq \lambda \cdot f_n\cdot \frac{w\br{A_{\mathcal{G}}\cup S_{\mathcal{G}}}\cdot w\br{B_{\mathcal{G}}\cup S_{\mathcal{G}}}}{w\br{\mathcal{G}}}\cdot c\br{S_{\mathcal{G}}^*} \leq 
        \lambda \cdot f_n\cdot w\br{B_{\mathcal{G}}\cup S_{\mathcal{G}}}\cdot c\br{S_{\mathcal{G}}^*} 
        \\&\leq 
        \lambda \cdot f_n\cdot \sum_{k=w\br{A_{\mathcal{G}}}+1}^{w\br{\mathcal{G}}}c\br{S_{\fl{k/2}}^*\cap \mathcal{G}}
        \end{align*}

        where the last inequality is by Lemma \ref{splitting}.

        As before, to bound the cost of the whole solution we will firstly show the following lemma. The argument is mostly the same as for the Lemma \ref{up_trees}, however, there are few differences and we include it for completeness:
            \begin{lemma}\label{up_graphs}
            $$\sum_{\mathcal{G}}\sum_{k=w\br{A_{\mathcal{G}}}+1}^{w\br{\mathcal{G}}}c\br{S_{\fl{k/2}}^*\cap \mathcal{G}}\leq \sum_{k=0}^{w\br{G}}c\br{S_{\fl{k/2}}^*}.$$
            \end{lemma}
            \begin{proof}
                Fix a value of $\mathcal{G}$ and $k$. Their contribution to the cost is $c\br{S_{\fl{k/2}}^*\cap \mathcal{G}}$. Consider which candidate subgraphs contribute such a term. By definition of $S_{\mathcal{G}}$, we have that $\mathcal{G}$ is the minimal subgraph, such that $w\br{\mathcal{G}}\geq k\geq w\br{A_{\mathcal{G}}}+1 >w\br{H}$, for every $H\in \mathcal{G}-S_{\mathcal{G}}$. This means that if for every $H\in \mathcal{G}-S_{\mathcal{G}}$, $w\br{H}<k$, then $\mathcal{G}$ contributes such a term. Since for all $H_1, H_2\in \mathcal{G}-S_{\mathcal{G}}$, $H_1\cap H_2=\emptyset$, we have that $\br{S_{\fl{k/2}}^*\cap H_1}\cup \br{S_{\fl{k/2}}^*\cap H_2}=\emptyset$, the claim follows by summing over all values of $k$.
            \end{proof}
            
        We are now ready to bound the cost of the solution. Let $D$ be the decision tree returned by the procedure. We have:
        \begin{align*}
            c_G\br{D}&\leq \sum_{\mathcal{G}}w\br{\mathcal{G}}\cdot c\br{S_{\mathcal{G}}}\leq \lambda \cdot f_n\cdot \sum_{\mathcal{G}}\sum_{k=w\br{A_{\mathcal{G}}}+1}^{w\br{\mathcal{G}}}c\br{S_{\fl{k/2}}^*\cap \mathcal{G}}
            \\&\leq 
            \lambda \cdot f_n\cdot\sum_{k=0}^{w\br{G}}c\br{S_{\fl{k/2}}^*} \leq 2\cdot \lambda \cdot f_n\cdot \OPT\br{G} = \br{12+4\sqrt{5}}\cdot f_n \cdot \OPT\br{G}
        \end{align*}
        
    where the third inequality is due to Lemma \ref{up_graphs} and the last inequality is by Lemma \ref{lb_opt}.
    \end{proof}

It should be noted that for multiple classes of graphs, better approximation algorithms for the Min-Ratio Vertex Cut Problem are known. For example if $G$ excludes $H$ as a minor for some fixed graph $H$, then an $O\br{\spr{V\br{H}}^2}$-approximation algorithm for the Min-Ratio Vertex Cut Problem is given in \cite{Improvedapproximationalgorithmsvertexseparators}. This immediately implies constant factor algorithms for multiple classes of graphs, for example planar graphs or graphs with bounded treewidth. Observe that the analysis of our algorithm loses a factor of $2\lambda=12+4\sqrt{5}\approx 20.94$ times $f_n$. We would like to bring $\lambda$ closer to 1. Below we show an improved version of Lemma \ref{lambda_lemma} showing that $\lambda$ can be reduced to roughly $5.95$.
\begin{lemma}
    Let $\mathcal{H}=\mathcal{G}-S_{\mathcal{G}}^*$ and let $\lambda=\frac{22+2\sqrt{5}+\br{1+\sqrt{5}}\cdot\sqrt{38+2\sqrt{5}}}{8}\approx5.95$. Then, we can partition $\mathcal{H}$ into two sets, $\mathcal{A}$ and $\mathcal{B}$ such that for $A=\bigcup_{H\in\mathcal{A}}V\br{H}$ and $B=\bigcup_{H\in\mathcal{B}}V\br{H}$, we have:
            $$w\br{A\cup S_{\mathcal{G}}^*}\cdot w\br{B\cup S_{\mathcal{G}}^*}\geq w\br{\mathcal{G}}^2/\lambda.$$
\end{lemma}
     \begin{proof}
        Let $x,y\in \mathbb{R}$ be two real numbers, such that $2<x\leq y$.
        There are three cases:
        \begin{enumerate}
            \item $w\br{ S_{\mathcal{G}}^*}\geq w\br{\mathcal{G}}/x$. In this case we take arbitrary partition $\mathcal{A}, \mathcal{B}$ of $\mathcal{H}$. We have:
                    $$w\br{A\cup S_{\mathcal{G}}^*}\cdot w\br{B\cup S_{\mathcal{G}}^*}\geq w\br{ S_{\mathcal{G}}^*}^2 \geq w\br{\mathcal{G}}^2/x^2.$$
            
            \item $w\br{\mathcal{G}}/y\leq w\br{ S_{\mathcal{G}}^*} \leq w\br{\mathcal{G}}/x$. For any choice of the partition $\mathcal{A},\mathcal{B}$ of $\mathcal{H}$, we have $\frac{w\br{A\cup B}}{w\br{\mathcal{G}}}\geq 1-\frac{1}{x}$. 
            We pick $\mathcal{A},\mathcal{B}$ to be a partition of $\mathcal{H}$, such that $w\br{A}\geq w\br{B}\geq \br{\frac{1}{2}-\frac{1}{x}}\cdot w\br{\mathcal{G}}$ (this is always possible as $\frac{1}{2}-\frac{1}{x}>0$ and for each $H\in\mathcal{H}$, $w\br{H}\leq w\br{\mathcal{G}}/2$).
            We have:
               \begin{align*}
                w\br{A\cup S_{\mathcal{G}}^*}\cdot w\br{B\cup S_{\mathcal{G}}^*}&\geq w\br{A}\cdot w\br{B} +w\br{S_{\mathcal{G}}^*}\cdot w\br{A\cup B}+w\br{S_{\mathcal{G}}^*}^2
                \\&\geq w\br{\mathcal{G}}^2\cdot \br{\frac{1}{4}-\frac{1}{2x}+\frac{1}{y}-\frac{1}{xy}+\frac{1}{y^2}}.
               \end{align*} 
                
            \item $w\br{S_{\mathcal{G}}^*} \leq w\br{\mathcal{G}}/y$. For any choice of the partition $\mathcal{A},\mathcal{B}$ of $\mathcal{H}$, we have $\frac{w\br{A\cup B}}{w\br{\mathcal{G}}}\geq 1-\frac{1}{y}$. We pick $\mathcal{A},\mathcal{B}$ to be a partition of $\mathcal{H}$, such that $w\br{A}\geq w\br{B}\geq \br{\frac{1}{2}-\frac{1}{y}}\cdot w\br{\mathcal{G}}$ (this is always possible as $\frac{1}{2}-\frac{1}{y} > 0$ and for each $H\in\mathcal{H}$, $w\br{H}\leq w\br{\mathcal{G}}/2$). We have:
                $$
                w\br{A\cup S_{\mathcal{G}}^*}\cdot w\br{B\cup S_{\mathcal{G}}^*}\geq w\br{A}\cdot w\br{B} \geq w\br{\mathcal{G}}^2\cdot \br{\frac{1}{4}-\frac{1}{2y}}.
                $$
        \end{enumerate}
        
        We then find the values $x, y$, such that $\frac{1}{x^2}=\frac{1}{4}-\frac{1}{2x}+\frac{1}{y}-\frac{1}{xy}+\frac{1}{y^2}=\frac{1}{4}-\frac{1}{2y}$. The only pair of solutions to this equation, for which $2\leq x\leq y$, is:
        $$
        \begin{cases}
            x=\frac{1}{4}\cdot \br{1+\sqrt{5}+\sqrt{38+2\sqrt{5}}}\approx 2.43828\\
            y=\frac{1}{4}\cdot \br{-1+3\sqrt{5}+\sqrt{38+2\sqrt{5}}}\approx 6.11263
        \end{cases}
        $$

        and one can easily check that $\lambda = x^2$, so the claim follows.
    \end{proof}

We get that the approximation ratio of our algorithm is at most $2\cdot\lambda\cdot f_n\approx 11.89\cdot f_n$.
    
\subsection{FPTAS for graphs with bounded number of non-leaf vertices}\label{subsec:graphs-parametrized}
In this section, we present an FPTAS for the Graph Search Problem running in time parameterized by the number of non-leaf vertices of the input graph $k$. To showcase the main ideas, we firstly present a simplified version of the FPTAS for stars ($k=1$). Then, we extend the result to general graphs with $k$ non-leaf vertices. 
In this section we treat the decision tree as a form of a schedule in which each query to $v\in V\br{G}$ is a job. After performing a query to $v$ in subgraph $G'$ the machine processing it is replaced by $\spr{G'-v}$ copies of that machine, each processing queries in one connected component of $G'-v$. In order to build our schedules effectively we will allow the decision tree to contain idle times between queries. By $C_v$ we shall denote the completion time of a query to $v$, which includes these idle times. Note that upon termination of our procedures, if such idle times are present in the resulting decision tree, we can always remove them without increasing the cost of the solution. 

\begin{theorem}
    For any $0<\epsilon \leq 3$ there exists an FPTAS for the Tree Search Problem on stars running in time $O\br{n^2\cdot \log^2 c\br{T}/\br{\epsilon\cdot\log\br{1+\epsilon}}}$ (Note that $\log\br{c\br{T}}$ is the number of bits required to represent $c$ so the running time is polynomial).
\end{theorem}
\begin{proof}
Let $r=r\br{T}$ be the center of the star and let $\br{v_1, v_2, \ldots, v_m}$ be its leaves. We devise a following dynamic program which finds an almost optimal decision tree for this case. Let $A[t]$ be the cost of the optimal decision tree assuming that $r$ is queried at the moment $t$. Ideally, to find the global optimum we would like to compute $\min_{t} A[t]$. However, since the only reasonable upper bound on $t$ is $c\br{T}$, we would only like to consider some of the values of $t$ to keep the running time polynomial. Let $\delta=\epsilon/3$. Consider the following set of moments $\mathcal{M}=\brc{0,1, (1+\delta)^1, (1+\delta)^2, \ldots, (1+\delta)^{\cl{\log c\br{T}/\log\br{1+\delta}}}}$. We have the following lemma:
    \begin{lemma}
        Let $t^*$ be the optimal moment to query $r$. There exists $t'\in \mathcal{M}$ such that $A[t']\leq \br{1+\delta}\cdot A[t^*]$.
    \end{lemma}
    \begin{proof}
        Assume that $t^*\neq 0$ and $t^*\neq \br{1+\delta}^{\cl{\log c\br{T}/\log\br{1+\delta}}}$, since otherwise we are done. Let $j$ be such that $(1+\delta)^{j-1}<t^*\leq (1+\delta)^j$. Let $D^*$ be the optimal decision tree assuming that $r$ is queried at the moment $t^*$. We will construct a decision tree $D'$ in which $r$ is queried at the moment $t'=(1+\delta)^j$ and show that for every vertex $x\in V\br{T}$, $c\br{D', x}\leq \br{1+\delta}\cdot c\br{D^*, x}$. The decision tree $D'$ is constructed as follows: We start with $D^*$ and delay the query to $r$ to the moment $t'=(1+\delta)^j$. This changes the contribution of $r$ and previous queries to the cost each vertex $v\in T_{D^*, r}$ from at least $\br{1+\delta}^{j-1}+c\br{r}$ to at most $\br{1+\delta}^{j}+c\br{r}$. Since $\frac{\br{1+\delta}^{j}+c\br{r}}{\br{1+\delta}^{j-1} + c\br{r}}\leq 1+\delta$ and the contribution of other vertices to the cost remains unchanged, the claim follows.
    \end{proof}

    Now we notice that for every $t\in \mathcal{M}$, $A[t]$ is an instance of the following problem known as \emph{Scheduling with Rejection}. In this problem we are given a set of jobs $[n]$ where each job $i$ has processing time $p_i$ and rejection cost $e_i$. The goal is to find a subset of jobs $S\subseteq [n]$ which are chosen to be scheduled. It is additionally assumed that all jobs which are scheduled finish at time at most $\tau$. The objective is to minimize $\sum_{i\in S}w_iC_i+\sum_{i\notin S}e_i$ where $C_i$ is the completion time of job $i$ in the schedule. For this problem an FPTAS running in $O\br{n^2/\delta\cdot\log \sum_{i}p_i}$ time exists \cite{scheduling_with_rejection}. 
    
    To see that $A[t]$ is a special case of scheduling with rejection, for every leaf $v_i$, we create job $i$ with $p_i=c\br{v_i}$, $w_i=w\br{v_i}$ and $e_i=w\br{v_i}\cdot\br{t+c\br{r}+c\br{v_i}}$. Additionally, we set $\tau=t$ and we create job $0$ with $p_0=\infty$ and $e_0=w\br{r}\cdot\br{t+c\br{r}}$. To see that these problems are equivalent notice that in any feasible solution of the scheduling with rejection instance, job $0$ must be rejected, which corresponds to the fact that $r$ is always queried at moment $t$. All accepted jobs are the queries preceeding the query of $r$ which are processed in a sequential manner. The cost of rejecting a job is the contribution of the leaves which are assumed to be queried after $r$ in the decision tree. Therefore, we can use the algorithm of \cite{scheduling_with_rejection} to compute an $\br{1+\delta}$-approximation of $A[t]$ for every $t\in \mathcal{M}$ in time $O\br{n^2/\delta\cdot\log c\br{T}}$. By choosing the best solution among all $t\in \mathcal{M}$ and using the above lemma, we get an $\br{1+\delta}^2\leq \br{1+3\delta}=\br{1+\epsilon}$-approximation. The running time is $O\br{\spr{\mathcal{M}}\cdot n^2\cdot\log c\br{T}/\delta}=O\br{n^2\cdot \log^2 c\br{T}/\br{\epsilon\cdot\log\br{1+\epsilon}}}$ as required.

\end{proof}

In order to extend the above result to graphs with a bounded number of non-leaf vertices, we use a similar idea, however we need a stronger dynamic program for the individual subproblems. To do so, we generalize the aforementioned FPTAS for scheduling with rejection.
\begin{theorem}
    For any $0<\epsilon\leq 1$ there exists an FPTAS for the Graph Search Problem for graphs with $k$ non-leaf vertices running in time $O\br{k\cdot n^{4k+2}\cdot \log^{2k} c\br{G}/\br{\log^{k}\br{1+\epsilon}}}$.
\end{theorem}
\begin{proof}
In order for the state space to have polynomial size, we will use the following aligning technique. Recall, that for the case of stars, we have only considered decision trees which query the center of the star at a moment which is a power of $\br{1+\delta}$. In what follows, we employ a similar idea, however this time we will require that each query ends at a moment which is a power of $\br{1+\delta}$ for some value of $\delta$ to be defined later. We say that a decision tree is \emph{aligned} if for every query in the decision tree, the query ends at a moment $\tau_j=\br{1+\delta}^j$. We have the following lemma:
\begin{lemma}
    Let $D^*$ be an optimal decision tree for the Graph Search Problem on a graph $G$. There exists an aligned decision tree $D'$ such that for every vertex $x\in V\br{G}$, $c\br{D', x}\leq \br{1+\delta}^n\cdot c\br{D^*, x}$.
\end{lemma}
\begin{proof}
    We construct $D'$ from $D^*$ as follows: We start with the root of $D^*$ and we process the decision tree in a top-down manner. Whenever we reach a query of a vertex $v$ which ends at a moment which is not a power of $\br{1+\delta}$, we delay the query of $v$ to the next such moment. Let $h$ be the depth of the query to $v$ in $D^*$, i.~e., its distance from the root (measured in the number of vertices). We show by induction on $h$ that after aligning queries up to $v$, for every vertex $x\in G_{D^*, v}$, $c\br{D', x}\leq \br{1+\delta}^{h}\cdot c\br{D^*, x}$. We notice that whenever a query finished at time $t\in \left(\tau_{i-1}, \tau_i\right]$ after all queries preceeding it, have already been aligned its completion time is changed to $\br{1+\delta}\cdot t$.
    This established the base case for $h=1$. Assume that the claim holds for some $h$. Let $v$ be a query at depth $h+1$ in $D^*$ and let $p$ be its parent. By induction, we have that after alignment, the completion time of $p$ increased by an additive factor of at most $[\br{1+\delta}^h-1]\cdot C_p$. This additive factor is also added to the completion time of $v$ before alignment. Therefore, after alignment, the completion time of $v$ is at most $\br{1+\delta}\cdot\br{C_v+[\br{1+\delta}^h-1]\cdot C_p}\leq \br{1+\delta}^{h+1}\cdot C_v$ where the inequality follows by the fact that $C_p\leq C_v$. Since the number of queries in $D^*$ is at most $n$, the claim follows.
\end{proof}

In what follows we set $\delta=\frac{\epsilon\cdot\ln2}{n}$. Observe that $2^x\leq 1+x$ for any $x\in [0,1]$. To see that this inequality is correct, notice that for $0$ and $1$ equality holds and the function $f(x)=2^x$ is convex and therefore always lesser than or equal to a linear function $1+x$ on the interval $[0,1]$. We get that:

$$
\br{1+\delta}^n=\br{1+\frac{\epsilon\cdot\ln2}{n}}^n\leq e^{\epsilon\cdot\ln2}=2^{\epsilon}\leq 1+\epsilon.
$$

Similarly, as for the star case we wish to guess the moments at which the queries to the inner vertices are finished. Since there are $k$ inner vertices, we will need to guess $k$ such moments. As the depth of the decision tree measured in the incurred cost is trivially upper bounded by $\br{1+\delta}^n\cdot c\br{G}$, we have that the number of possible moments to consider is at most $\cl{n\cdot\log c\br{G}/\log\br{1+\delta}}=O\br{n^2\cdot\log c\br{G}/\log\br{1+\epsilon}}$.
Let $\mathcal{M}$ be the set of all $k$-tuples $\br{\tau_{v_1}, \tau_{v_2}, \ldots, \tau_{v_k}}$ denoting the moments at which the queries to the inner vertices $v_1,\dots,v_k$ are finished in an aligned decision tree. Any such tuple encodes a precise order in which the inner vertices are queried, however we will be only interested in tuples which dictate a valid order of queries which can be easily checked in polynomial time. We have that $\spr{\mathcal{M}}=O\br{n^{2k}\cdot\log^k c\br{G}/\br{\log^k\br{1+\delta}}}$.

For every valid $M\in \mathcal{M}$ we will compute an entry $A[M]$ in which the queries to the inner vertices are already scheduled according to $M$ and then we will choose the best among $\mathcal{M}$. Let $D_M$ be a decision tree consisting of non-leaf vertices as dictated by $M$. The remaining task is to schedule queries to the leaves of $G$ in between queries in $D_M$. To do so, our dynamic programming will store information about the occupation of the time intervals in-between each two consecutive inner queries (as dictated by $M$). More precisely, let $uv\in E\br{D_M}$ be two consecutive inner vertices in the order dictated by $M$. For each such pair we will store an index $\tau_{uv}$ denoting that every query to a leaf in between queries to $u$ and $v$ ends at time at most $\tau_{uv}$. Let $\mathcal{T}$ be the set of such indices. We process the leafs queries $l_1,\dots, l_{n-k}$ according to the Smiths rule, i. e., in an ascending order of $c\br{l_i}/w\br{l_i}$. We define the following dynamic programming state: $B[i, \mathcal{T}]$ is the cost of the optimal decision tree when the following conditions hold:
\begin{itemize}
    \item The queries to the leaves $l_1, l_2, \ldots, l_i$ have already been scheduled.
    \item For each pair of consecutive inner vertices $uv\in E\br{D_M}$, every query to a leaf in between queries to $u$ and $v$ ends at time at most $\tau_{uv}\in \mathcal{T}$.
\end{itemize}

Let $r=r\br{D_M}$ be the root of $D_M$. 
For every $l_i$, let $v\br{l_i}$ be the inner vertex such that $l_i$ is connected to $v\br{l_i}$ in $G$. We start by giving the boundary case when $i=1$. We have the following recurrence:
\begin{align*}
&B[1, \mathcal{T}] = w\br{l_1}\cdot\min\brc{\tau_{v\br{l_1}}+c\br{l_1}, \min_{uv \in P_{r,v\br{l_1}}\br{D_M}}\brc{\infty, \min_{\tau_u+c\br{l_1}\leq \tau \leq \tau_{u,v}, \tau = \br{1+\delta}^p, p\in \mathbb{N}}\brc{\tau}}} 
\end{align*}
with the following cases to consider:
\begin{enumerate}
    \item The first case corresponds to the situation when we choose $l_1$ to be queried after $v\br{l_1}$. If so, the contribution of $l_1$ is exactly $w\br{l_1}\cdot\br{\tau_{v\br{l_1}}+c\br{l_1}}$.
    \item The second case corresponds to the situation when we choose to query $l_1$ between some consecutive queries $u,v$ in the path from $r$ to $v\br{l_1}$ in $D_M$ (since these are the only possible places to query $l_1$). In this case we need to find the earliest possible moment $\tau$ in which we can schedule the query to $l_1$ such that $\tau\leq \tau_{u,v}$ and $\tau\geq \tau_u+c\br{l_1}$. Note that since we are interested in aligned decision trees, we only consider moments $\tau$ which are powers of $\br{1+\delta}$.
\end{enumerate}

Now we show how to compute $B[i, \mathcal{T}]$ for $i>1$. Let $uv\in E\br{D_M}$, $\tau_{u,v}'$ be the largest aligned moment such that $\tau_{u,v}'+c\br{l_i}\leq \tau_{u,v}$ and $\mathcal{T}_{u,v}'$ be $\mathcal{T}$ where $\tau_{u,v}$ is replaced by $\tau_{u,v}'$. We have the following recurrence:
\begin{align*}
&B[i, \mathcal{T}] = \min\brc{w\br{l_i}\cdot\left[\tau_{v\br{l_i}}+c\br{l_i}\right]+B[i-1, \mathcal{T}], \min_{uv \in P_{r,v\br{l_i}}\br{D_M}}\brc{w\br{l_i}\cdot\tau_{uv}+B[i-1, \mathcal{T}_{u,v}']}}
\end{align*}
again, with the following cases to consider:
\begin{enumerate}
    \item The first case corresponds to the situation when we choose $l_i$ to be queried after $v\br{l_i}$. If so, we know that $l_i$ contributes exactly $w\br{l_i}\cdot\left[\tau_{v\br{l_i}}+c\br{l_i}\right]$ to the objective function and the remaining cost is $B[i-1, \mathcal{T}]$.
    \item The second case corresponds to the situation when we choose to query $l_i$ between some consecutive queries $u,v$ in the path from $r$ to $v\br{l_i}$ in $D_M$. In this case, by Smith's rule we know that among leaf queries chosen to be performed between queries to $u$ and $v$ the query to $l_i$ is the last one. Otherwise we could swap the queries to arrive at a decision tree of strictly lower cost. Therefore, we may assume that query to $l_i$ finishes at time $\tau_{uv}$ so its contribution to the objective value is $w\br{l_i}\cdot\tau_{uv}$. The rest of the objective value is the cost of scheduling the first $i-1$ leaf queries with the same constraints as before, however with $\tau_{uv}$ replaced by $\tau_{u,v}'$.
\end{enumerate}

Observe that $A[M]=B[n-k, \mathcal{T}_M]$, where in $\mathcal{T}_M$ for every $uv\in E\br{D_M}$ we set $\tau_{uv}=\tau_{v}$. By choosing the best $M$ among all valid $M\in \mathcal{M}$ and using the alignment lemma we get an $\br{1+\epsilon}$-approximation as required.
\end{proof}

\begin{lemma}
The running time of the algorithm is $O\br{k\cdot n^{4k+2}\cdot \log^{2k} c\br{G}/\br{\log^k\br{1+\delta}}}$.
\end{lemma}
\begin{proof}
    Clearly, $\spr{\mathcal{T}}\leq k$. Therefore, there are $O\br{n^{2k}\cdot \log^{k} c\br{G}/\br{\log^k\br{1+\delta}}}$ possible values of $\mathcal{T}$, $n$ possible values of $i$ and $\spr{\mathcal{M}}$ possible values of $M$. Each subproblem can be computed in time $O\br{kn}$. Multiplying these values gives the required running time.
\end{proof}

\section{Trees, monotone cost functions}\label{sec:Average-case searching}

Below we show an integer linear programming formulation of the separator assignment. This will allow us to build a $2$-approximation algorithm for the case when the cost function is monotone, non-decreasing with respect to the target. The ILP is inspired by the approach of \cite{Angelidakis2018ShortestPQ} for the Hub Labeling Problem which for uniform costs is equivalent to finding a sequence assignment with small cost. 
We introduce a binary variable $x_{u,v}$ for every ordered pair of vertices $u,v\in V\br{T}$, which indicates whether $u\in S\br{v}$. We also introduce binary variables $y_{u,v,w}$ for every ordered triple of vertices $u,v,w\in V\br{T}$, which indicates whether $u\in S\br{v}\cap S\br{w}$. Note that $y_{u,v,w}$ and $y_{u,w,v}$ encode the same event, however for the sake of simplicity of writing we include both variables. The ILP is as follows:

\begin{align*}
\text{min} \quad & \sum_{v\in V\br{T}}w\br{v}\cdot\sum_{u\in V\br{T}} c\br{u,v}\cdot x_{u,v} \\
\text{subject to} \quad & y_{u,v,w} \leq x_{u,v} & \forall u,v,w\in V\br{T}\\
& y_{u,v,w} \leq x_{u,w} & \forall u,v,w\in V\br{T}\\
& \sum_{u\in P_{v,w}} y_{u,v,w} \geq 1 & \forall v,w\in V\br{T}\\
& x_{u,v} \in \brc{0,1} & \forall u,v\in V\br{T}\\
& y_{u,v,w} \in \brc{0,1} & \forall u,v,w\in V\br{T}
\end{align*}

The objective function is to minimizes the min sum cost of the separator assignment. The first two constraints ensure that $y_{v,u,w}$ can only be non-zero if both $x_{v,u}$ and $x_{w,u}$ are non-zero, which corresponds to $u$ being in both $S\br{v}$ and $S\br{w}$. The third constraint ensures that for any pair of vertices $v$ and $w$, there exists at least one vertex $u$ on the path between them that is both in
$S\br{v}$ and $S\br{w}$, which corresponds to the separator assignment property.
\subsection{The algorithm}
\label{subsec:avg-algorithm}
\SetKwFunction{FDecisionTree}{LPBasedDecisionTree}
\SetKwFunction{FPseudoSepAssignment}{PseudoSepAssignment}
\SetKwFunction{FReconstructDecisionTree}{ReconstructDecisionTree}
The Algorithm \ref{alg:AverageCaseDecisionTree} works as follows: We first solve the linear programming relaxation of the above ILP. Then, for each vertex $v\in V\br{T}$, we create a tree $T^v$ by rerooting $T$ at $v$. We then use the \FPseudoSepAssignment procedure on $T^v$ and the obtained fractional variables $\mathbf{x}$ to construct a pseudo-separator assignment $S\br{v}$ in a bottom-up manner. Notably, despite the fact that for each vertex $v\in V\br{T}$ the assignment $S\br{v}$ is calculated independently, the resulting collection of sets meets the criteria of the pseudo-separator assignment. Finally, we use this assignment $S$ to reconstruct a decision tree in a top-down manner using the \FReconstructDecisionTree procedure. The details of each step are provided below.

\begin{algorithm}[H]
\caption{The 2-approximation algorithm for the Average-Case Tree Search Problem.}
\label{alg:AverageCaseDecisionTree}
\SetKwProg{Fn}{Procedure}{:}{}
\SetKwFor{For}{\textup{for}}{\textup{do}}{}
\Fn{$\FDecisionTree\br{T,c,w}$}{
Let $\mathbf{x}, \mathbf{y}$ be the solutions to the following LP:
\begin{align*}
\text{min} \quad & \sum_{v\in V\br{T}}w\br{v}\cdot\sum_{u\in V\br{T}} c\br{u,v}\cdot x_{u,v} \\
\text{s.t.} \quad & y_{u,v,w} \leq x_{u,v} & \forall u,v,w\in V\br{T}\\
& y_{u,v,w} \leq x_{u,w} & \forall u,v,w\in V\br{T}\\
& \sum_{u\in P_{v,w}} y_{u,v,w} \geq 1 & \forall v,w\in V\br{T}\\
& x_{u,v} \in [0,1] & \forall u,v\in V\br{T}\\
& y_{u,v,w} \in [0,1] & \forall u,v,w\in V\br{T}
\end{align*}

\For{$v\in V\br{T}$}{
    $T^v \gets$ $T$ rerooted at $v$.

    $S\br{v} \gets \FPseudoSepAssignment\br{T^v, v, \mathbf{x}}$.
}

$D\gets \FReconstructDecisionTree\br{T, S}$.

\Return $D$.
}
\end{algorithm}
\subsection{Constructing a pseudo-separator assignment}
\label{subsec:constructing-pseudo-separator}
\begin{lemma}
\label{pseudo_separator_assignment_lemma}
There exists a polynomial time algorithm \FPseudoSepAssignment which, given a tree $T$, root $v$, and a feasible solution to the LP relaxation $\mathbf{x}$, constructs a pseudo-separator assignment $S\br{v}$ such that $c\br{S\br{v}, v}\leq 2\cdot \sum_{u\in V\br{T}} c\br{u,v}\cdot x_{u,v}$.
\end{lemma}
\begin{proof}

The Algorithm \ref{alg:pseudoSepAssignment} is recursive and takes as input a tree $T$, a vertex $v$ for which the assignment $S\br{v}$ is to be constructed (not necessarily belonging to $V\br{T}$) and the fractional solution $\mathbf{x}$. Note that since we update the values of $\mathbf{x}$ during the execution of the algorithm, this variable is passed by reference. The output of the algorithm is the $S\br{v}\cap V\br{T}$ subset of $S\br{v}$. The whole solution $S\br{v}$ is constructed for $T^v$ and $v$. The algorithm works in a bottom up manner. It starts by initializing $S\br{v}=\emptyset$. Then, for each $H \in T-r\br{T}$, it recursively calls itself on $H$ and $v$ and each set returned by these calls is added to $S\br{v}$. Let $x\br{T, v} = \sum_{u \in V\br{T}}x_{u,v}$. After recursively processing all subtrees, the algorithm checks whether $x\br{T, v}\geq 1/2$. If this is the case, then $r\br{T}$ is added to $S\br{v}$ as well and for each $u\in V\br{T}$, $x_{u,v}$ is updated to be $0$. Then the algorithm returns the resulting assignment $S\br{v}\cap V\br{T}$. 
\end{proof}

\begin{algorithm}[H]
\caption{The algorithm to construct a pseudo-separator assignment. Note that $\mathbf{x}$ is passed by reference.}
\label{alg:pseudoSepAssignment}
\SetKwFunction{FDecisionTree}{DecisionTree}
\SetKwProg{Fn}{Procedure}{:}{}
\SetKwFor{For}{\textup{for}}{\textup{do}}{}
\Fn{$\FPseudoSepAssignment\br{T, v, \mathbf{x}}$}{
$S\br{v}\gets \emptyset$.

\For{$H \in T-r\br{T}$}{
    $S\br{v}\gets S\br{v} \cup \FPseudoSepAssignment\br{H, v, \mathbf{x}}$.
}
\If{$x\br{T, v} \geq \frac{1}{2}$}{
    $S\br{v} \gets S\br{v} \cup \{r\br{T}\}$.
    
    \For{$u \in V\br{T}$}{
        $x_{uv} \gets 0$.
    }
}
\Return $S\br{v}$.
}
\end{algorithm}
\begin{lemma}
\label{pseudo_separator_assignment_correctness_lemma}
The assignment $S$ built by iteratively calling \FPseudoSepAssignment for each vertex $v\in V\br{T}$ is a pseudo-separator assignment.
\end{lemma}
\begin{proof}
Consider any pair $u,v \in V\br{T}$. Let $u'$ be the vertex in $S\br{u}$ which lies on the path $P_{u,v}$ closest to $v$ and let $v'$ be the vertex in $S\br{v}$ which lies on the path $P_{u,v}$ closest to $u$. Assume towards a contradiction that $u'\notin P_{v, v'}$. Also assume that $v'\neq u$ and $u'\neq v$ as otherwise the claim holds trivially. Let $u''$ be the first vertex on the path $P_{u',v}$ after $u'$ and symetrically, let $v''$ be the first vertex on the path $P_{v',u}$ after $v'$. Since none of the vertices on path $P_{v'',u}$ belongs to $S\br{v}$ it must be the case that $\sum_{w \in P_{v'',u}} x_{w,v}< 1/2$, otherwise the algorithm would have added one of these vertices to $S\br{v}$. Symetrically, we also have that $\sum_{w \in P_{u'',v}} x_{w,u}< 1/2$. Therefore, we get:
\begin{align*}
\sum_{w \in P_{u,v}} y_{w, u, v} &\leq \sum_{w \in P_{u,v}} \min\brc{x_{w,u}, x_{w,v}} 
\leq 
\sum_{w \in P_{w\in P_{u''v}}}x_{w,u} + \sum_{w \in P_{w\in P_{v''u}}}x_{w,v}
< 1/2 + 1/2 = 1
\end{align*}

which is a contradiction. The claim for $v'$ is proven symmetrically.
\end{proof}
\begin{lemma}
\label{pseudo_separator_assignment_cost_lemma}
For every vertex $v \in V\br{T}$, the algorithm \FPseudoSepAssignment constructs a pseudo-separator assignment $S\br{v}$ such that $c\br{S\br{v}, v}\leq 2\cdot \sum_{u\in V\br{T}} c\br{u,v}\cdot x_{u,v}$.
\end{lemma}
\begin{proof}
For any $w \in S\br{v}$ let $H_w^v\subseteq V\br{T}$, be the set of vertices $u$ such that at the moment of appending $w$ to $S\br{v}$, $x_{u, v}>0$. Let $T_w^v$ be the subtree of $T$ induced by $H_w^v$. We have:

\begin{align*}
c\br{S\br{v}, v} 
&= 
\sum_{w\in S\br{v}} c\br{w,v} 
\leq 
2\cdot \sum_{w\in V\br{T}} c\br{w,v}\cdot \sum_{u \in H_w^v} x_{u,v} 
\\&\leq 
2\cdot \sum_{w\in V\br{T}} \sum_{u \in H_w^v} c\br{u,v} \cdot x_{u,v}
\leq 
2 \cdot \sum_{u\in V\br{T}} c\br{u,v}\cdot x_{u,v}. 
\end{align*}

where the first inequality is by the fact that $x\br{T_w^v, u} \geq 1/2$ for every $w\in S\br{v}$, the second inequality is by the monotonicity of $c$, since every $u\in H_w^v$ is a descendant of $w$ in $T^v$ and the last inequality is by the fact that the sets $H_w^v$ are disjoint for different $w$. This concludes the proof.
\end{proof}
\subsection{Reconstructing the decision tree}
\label{subsec:reconstructing-decision-tree}

\begin{lemma}
\label{reconstruct_decision_tree_lemma}
Given a tree $T$ and a pseudo-separator assignment $S$, there exists a polynomial time algorithm \FReconstructDecisionTree which constructs a decision tree $D$ such that for every $x\in V\br{T}$, $c_T\br{D,x} \leq c\br{S\br{x}, x}$.
\end{lemma}
\begin{proof}
The algorithm is recursive and takes as an input a tree $T$ and a pseudo-separator assignment $S$ and outputs a decision tree $D$. The algorithm works in a bottom-up manner. It starts by constructing a tree $T^{S\br{v}} = \bigcup_{u \in S\br{v}} P_{v,u}$ for each vertex $v$. By definition of the pseudo-separator assignment, the resulting family of trees satisfies the \textit{Helly} property \cite{GyarfasLehel1970}, i.e., for every $u,v \in V\br{T}$, we have that $V\br{T^{S\br{v}}} \cap V\br{T^{S\br{u}}} \neq \emptyset$. Therefore, there exists a vertex $r\in \bigcap_{v \in V\br{T}} V\br{T^{S\br{v}}}$. The root of $D$ is set to be $r$. Then, for each connected component $H \in T-r$, the subtree of $D$ corresponding to $H$ is built recursively and hanged below $r$ in $D$. Finally, the resulting decision tree $D$ is returned.

\begin{algorithm}[H]
\caption{The algorithm to reconstruct a decision tree from a pseudo-separator assignment.}
\label{alg:reconstructDecisionTree}
\SetKwFunction{FDecisionTree}{DecisionTree}
\SetKwFunction{FSeparatorFPTAS}{SeparatorFPTAS}
\SetKwProg{Fn}{Procedure}{:}{}
\SetKwFor{For}{\textup{for}}{\textup{do}}{}
\Fn{$\FReconstructDecisionTree\br{T, S}$}{
\For{$v \in V\br{T}$}{
    $T^{S\br{v}}\gets \bigcup_{u \in S\br{v}} P_{v,u}$.
}

$r\gets$ an arbitrary node in $\bigcap_{v \in V\br{T}} T^{S\br{v}}$.

$D\gets \br{r, \emptyset}$.

\For{$H \in T-r$}{
    $D_H \gets \FReconstructDecisionTree\br{H, S}$.

    Hang $D_H$ below $r$ in $D$.
}
\Return $D$.
}
\end{algorithm}

We show that for every $x\in V\br{T}$, $c_T\br{D,x} \leq c\br{S\br{x}, x}$. Consider any $x\in V\br{T}$. Since $r\in V\br{T^{S\br{x}}}$, there exists a vertex $s\in S\br{x}$ such that $s\in P_{r,x}$. If $r=s$, then the cost of the query sequence of $x$ built based on $S\br{x}$ remains unchanged. Otherwise, notice that after query to $r$, $x$ and $s$ belong to different connected components of $T-r$. This in turn implies that $s$ will never be queried when searching for $x$ and by the monotonicity of costs we have that $c\br{r,x} \leq c\br{s,x}$. Therefore, the cost of the query sequence of $x$ built based on $S\br{x}$ does not increase. Since the same also holds inductively for the subtrees of $D-r$, we get that $c_T\br{D,x} \leq c\br{S\br{x}, x}$ as required.
\end{proof}

\begin{theorem}
There exists a 2-approximation algorithm for the Average-Case Tree Search Problem running in polynomial time.
\end{theorem}
\begin{proof}
    By combining Lemmas \ref{pseudo_separator_assignment_lemma} and \ref{reconstruct_decision_tree_lemma}, we get that for every $v\in V\br{T}$, the cost of the query sequence of $v$ in the decision tree $D$ constructed by Algorithm \ref{alg:AverageCaseDecisionTree} is at most $2\cdot \sum_{u\in V\br{T}} c\br{u,v}\cdot x_{u,v}$. Therefore, the cost of $D$ is at most:
    $$
    \sum_{v\in V\br{T}}w\br{v}\cdot c_T\br{D,v} \leq 2\cdot \sum_{v\in V\br{T}}w\br{v}\cdot \sum_{u\in V\br{T}} c\br{u,v}\cdot x_{u,v}\leq 2\cdot \OPT\br{T}.
    $$
\end{proof}
\subsection{Worst-case searching}
\label{sec:worst-case-searching}
Using a very similar approach we can also obtain a $2$-approximation algorithm for the worst-case criterion. The only difference is that instead of solving the LP relaxation of the ILP presented in Section \ref{sec:Average-case searching}, we use the following ILP:
\begin{align*}
\text{min} \quad & M \\
\text{subject to} \quad & \sum_{u\in V\br{  T}} c\br{u,v}\cdot x_{u,v} \leq M & \forall v\in V\br{T}\\
& y_{u,v,w} \leq x_{u,v} & \forall u,v,w\in V\br{T}\\
& y_{u,v,w} \leq x_{u,w} & \forall u,v,w\in V\br{T}\\
& \sum_{u\in P_{v,w}} y_{u,v,w} \geq 1 & \forall v,w\in V\br{T}\\
& x_{u,v} \in \brc{0,1} & \forall u,v\in V\br{T}\\
& y_{u,v,w} \in \brc{0,1} & \forall u,v,w\in V\br{T}
\end{align*}

As a corollary we obtain the following result.
\begin{theorem}
There exists a 2-approximation algorithm for the Worst-Case Tree Search Problem running in polynomial time.
\end{theorem}
\begin{proof}
By combining Lemmas \ref{pseudo_separator_assignment_lemma} and \ref{reconstruct_decision_tree_lemma}, we get that for every $v\in V\br{T}$, the cost of the query sequence of $v$ in the decision tree $D$ constructed by the algorithm is at most $2\cdot \sum _{u\in V\br{T}} c\br{u,v}\cdot x_{u,v}$. Therefore, the worst-case cost of $D$ is at most:
$$
\max_{v\in V\br{T}} c_T\br{D,v} \leq 2\cdot \max_{v\in V\br{T}} \sum_{u\in V\br{T}} c\br{u,v}\cdot x_{u,v}\leq 2\cdot \OPT\br{T}.
$$
\end{proof}

\section{Hardness}\label{sec:hardness}
\subsection{Trees, costs nondependent on target}\label{subsec:hardness-trees-nondep}
We begin by showing that the Tree Search Problem with target independent cost functions is NP-complete. The decision version of the Tree Search Problem  is to determine whether, 
for a given instance $\br{T, c, w}$, there exists a decision tree of cost at most $K$.
\begin{theorem}
    The Decision Tree Search Problem is NP-complete even when restricted to trees with $\Delta\br{T}\leq 16$ and to trees with $\diam\br{T}\leq 6$.
\end{theorem}

    \begin{proof}
        The problem is in NP since, given a decision tree $D$, one can verify in polynomial time 
whether all the requirements are satisfied.

To show hardness, we use a black-box reduction from the edge-query, uniform-cost, 
and non-uniform-weight variant, which is NP-complete even when restricted to trees with $\Delta\br{T}\leq 16$ and to trees with $\diam\br{T}\leq 4$ \cite{Jacobs2010OnTheComplexSearchInTsAvg}. 
Let $\br{T, w, K}$ be such an instance. 
We construct a new instance $\br{T', c, w, K'}$ for the Tree Search Problem as follows: 
for every $v \in V\br{T}$, we set $c\br{v} = K+1$. 
If neither of the edge endpoints is a leaf, we subdivide each edge $e \in E\br{T}$ by adding a new vertex $v_e$ with 
$w\br{v_e} = 0$ and $c\br{v_e} = 1$. If otherwise, we set $w\br{v_e} = w\br{u}$ where $u$ is the leaf endpoint of $e$ and set $c\br{v_e} = 1$. Let $w_{\ell}\br{T} = \sum_{u \text{ is a leaf}} w\br{u}$.
We set $K' = K + \br{w\br{T}-w_{\ell}\br{T}} \cdot \br{K+1}$.

Assume that we have a decision tree $D$ of cost at most $K$ for the original instance. 
To obtain a decision tree $D'$ for the new instance, we replace each query in $D$ 
with a query to the vertex that subdivides the corresponding edge (or the leaf endpoint if the edge has one). 
Additionally, (if needed) below each leaf of $D$, we attach the appropriate queries to the original vertices. 
As $D$ contains a query to every edge of $T$, each vertex is separated, so for every 
$v \in V\br{T}$, one such additional query is added whenever $v$ is not a leaf. 
This results in a decision tree $D'$ of cost at most $K + \br{w\br{T}-w_{\ell}\br{T}} \cdot \br{K+1}= K'$, 
as required.

Observe that in the new instance, for every non leaf vertex $v \in T$, the cost of searching for $v$ 
is at least $K+1$ since $c\br{v} = K+1$. 
Therefore, for these vertices, at least $\br{w\br{T}-w_{\ell}\br{T}} \cdot \br{K+1}$ cost is required. 
This implies that each such vertex has exactly one such query in its query sequence, 
namely the query to $v$ itself. Otherwise, the cost would exceed $K'$, and we conclude 
that every such $v$ is queried only when the candidate subset consists solely of $v$. 
Assume there exists a decision tree $D'$ of cost at most $K'$ for the new instance. 
We show how to obtain a decision tree $D$ for the original instance. 
We replace each query to a vertex $v_e$ with a query to edge $e$, 
and delete all queries to vertices $v \in V\br{T}$. 
Since each $v \in V\br{T}$ was the last query in $Q_{D'}\br{T', v}$, performed 
when the candidate set consisted only of $v$, the resulting $D$ is a valid decision tree 
for the original instance. Additionally, the cost of $D$ is at most $K' - \br{w\br{T}-w_{\ell}\br{T}} \cdot \br{K+1} = K$, 
as required.

Finally, note that subdividing each edge (except those containing leaf vertices) results in a tree with diameter $\diam\br{T'}=2\cdot\diam\br{T}-2$ of the tree while leaving the largest degree unchanged, 
so the claim follows.

    \end{proof}

\subsection{Stars, costs dependent on target}\label{subsec:hardness-stars-dep}

Assume that the cost is allowed to be an arbitrary function $c\colon V\br{T}\times V\br{T} \to \mathbb{N}$. It should be noted that in this most general variant, the weight function is redundant since for every $u, v \in V\br{T}$, one can set $c'\br{v, u}=w\br{u}\cdot c\br{v, u}$ and therefore we can omit the weight function from the input.

The Unique Games Conjecture (UGC) introduced by Khot \cite{UGConjecture} is a statement about the hardness of a certain type of constraint satisfaction problems. An instance $\mathcal{L}=\br{G\br{V,W,E}, [n], \brc{\pi_{vw}}_{vw\in E}}$ of UGC consists of a regular bipartite graph $G\br{V,W,E}$ and a set of viable labels $[n]$. Each edge $vw \in E$ is associated with a permutation $\pi_{v,w}\colon [n] \to [n]$. A labeling is a function $f\colon V\cup W \to [n]$ assigning a label to each vertex. An edge $vw$ is satisfied by labeling $f$ if $\pi_{v,w}\br{f\br{v}} = f\br{w}$. The goal is to find a labeling that maximizes the fraction of satisfied edges. The UGC states that for every $\zeta, \gamma > 0$, there exists an integer $n=n\br{\zeta, \gamma}$ such that it is NP-hard to distinguish between the following two cases for an instance of UGC with label set $[n]$:
\begin{itemize}
    \item \textbf{YES Case:} There exists a labeling satisfying at least $1-\zeta$ fraction of edges.
    \item \textbf{NO Case:} Every labeling satisfies at most $\gamma$ fraction of edges.
\end{itemize}

The Minimum Cost Feedback Arc Set Problem is the following: Given a directed graph $G\br{V,E}$ and a cost funtion $c\colon E \to \mathbb{N}$, find a subset of edges $E'\subseteq E$ of minimum size such that the graph $\br{V, E\setminus E'}$ is acyclic. By \cite{BeatingTheRandomOrderingIsHardEveryOrderingCSPIsApproximationResistant}, the problem does not admit any constant factor approximation under the UGC. The Minimum Linear Ordering Problem is the following: Given a matrix $A\in \mathbb{N}^{n\times n}$, find a permutation $\pi$ of $[n]$ minimizing $\sum_{i<j} A_{\pi\br{i}, \pi\br{j}}$. In fact the Minimum Linear Ordering Problem is equivalent to the Minimum Cost Feedback Arc Set Problem \cite{OnApproximabilityOfLinearOrderingAndRelatedNPOptimizationProblemsOnGraphs}, which yields the same hardness of approximation result for the latter.
We show that Minimum Linear Ordering is a special case of the Average-Case Tree Search Problem with arbitrary costs.

\begin{theorem}
    The Average-Case Tree Search Problem with arbitrary costs is NP-hard and does not admit any constant factor approximation under the UGC even when the input tree is a star.
\end{theorem}
\begin{proof}
    Given an instance $A\in \mathbb{N}^{n\times n}$ of the Minimum Linear Ordering Problem, we construct an instance $\br{T, c, w}$ of the Average-Case Tree Search Problem as follows: We start with a vertex $r$ which will be the center of the star. For every $1\leq i \leq n$, we create a leaf $v_i$ connected to $r$. For every $i, j \in [n], i\neq j$, we set: $c\br{v_i, v_j} = A_{i,j}$, $c\br{r, v_i} = \infty$ for every $i\in[n]$ and $c\br{v_i, r} = 0$ for every $i\in[n]$, and $c\br{v, v} = 0$ for every $v \in V\br{T}$.

    It is easy to observe, that any reasonable strategy firstly queries all leafs (in some order) and then queries $r$. Additionally, notice that the contribution of $r$ to the objective function in such scenario is 0. Now, see that a given query results in a cost of $A_{i,j}$ if and only if the query to $v_i$ is performed before the query to $v_j$. Therefore, the cost of the decision tree is exactly $\sum_{i<j} A_{\pi\br{i}, \pi\br{j}}$ where $\pi$ is the permutation of $[n]$ corresponding to the order of queries to leafs. This concludes the proof.
\end{proof}

The state of the art algorithm for the Minimum Feedback Arc Set Problem achieves an approximation ratio of $O\br{\log n\cdot \log \log n}$ \cite{ApproximatingMinimumFeedbackSetsAndMulticutsInDirectedGraphs}. This suggests that the introduction of the target dependent costs makes the problem significantly harder to approximate. We leave as an open question whether any polylogarithmic approximation is achievable beyond paths, for which one can easily easily find the optimal solution using dynamic programming, since for $P=\angl{v_1,\dots, v_n}$, $\OPT[P]=\min_{1\leq i\leq n}\brc{\sum_{j=1}^{n}c\br{v_i, v_j}+\sum_{P'\in P-v_i}\OPT[P']}$ which can be solved in $O\br{n^3}$ time.

\bibliographystyle{plainnat}
\bibliography{references}

\end{document}